\documentclass[journal]{IEEEtran}
\IEEEoverridecommandlockouts
\usepackage{graphicx,cite}
\usepackage{textcomp}
\def\BibTeX{{\rm B\kern-.05em{\sc i\kern-.025em b}\kern-.08em
    T\kern-.1667em\lower.7ex\hbox{E}\kern-.125emX}}
\usepackage[utf8]{inputenc}
\usepackage{textcomp}
\DeclareRobustCommand{\sectsign}{\textsection}
\usepackage{url}
\usepackage{amsmath, amssymb, mathtools, amsthm}
\usepackage{multirow}
\usepackage{array}
\usepackage{booktabs} 
\usepackage{multicol}
\usepackage[shortlabels]{enumitem}
\usepackage{siunitx}
\usepackage{cancel}
\usepackage{graphicx}
\usepackage{pgfplots}
\usepackage{listings}
\usepackage{tikz}
\usepackage{bm}
\usepackage{subfigure}
\usepackage{algpseudocode}
\usepackage{algorithm}
\usepackage{lettrine}
\usepackage{makecell}
\newtheorem{theorem}{Theorem}
\newtheorem{lemma}{Lemma}

\allowdisplaybreaks
\usepackage[final]{changes}

\usepackage[colorlinks=true, allcolors=blue]{hyperref}

\usepackage{orcidlink}
\usepgfplotslibrary{external}
\title{Joint Scheduling of DER under Demand Charges: Structure and Approximation}

\makeatletter

\newcommand{\linebreakand}{%
      \end{@IEEEauthorhalign}
      \hfill\mbox{}\par
      \mbox{}\hfill\begin{@IEEEauthorhalign}
    }
    \makeatother
\author{%
\IEEEauthorblockN{Ruixiao Yang\textsuperscript{*}\orcidlink{0000-0002-2852-0150}, Gulai Shen\textsuperscript{*}\orcidlink{0000-0002-6016-0159}, Ahmed S. Alahmed\orcidlink{0000-0002-4715-4379}, Chuchu Fan\orcidlink{0000-0003-4671-233X}}

\thanks{\textsuperscript{*}Equal contribution. Ruixiao Yang (\textcolor{blue}{ruixiao@mit.edu}) and Chuchu Fan (\textcolor{blue}{chuchu@mit.edu}) are with the Massachusetts Institute of Technology, Cambridge, USA. Ahmed S. Alahmed (\textcolor{blue}{aalahmad@kfupm.edu.sa}) is with the King Fahd University of Petroleum and Minerals, Dhahran, KSA.
Gulai Shen (\textcolor{blue}{gulai\_shen@gsd.harvard.edu}) is with Harvard University, Cambridge, USA.}
\thanks{A preliminary version of this paper was published as a conference paper
at the 2025 IEEE Power \& Energy Society General Meeting (PESGM), Austin, TX, USA, July 2025 \cite{Yang&Shen&Alahmed&Fan:25PESGM}. ({\em Corresponding author: Ahmed S. Alahmed.})}}

\date{}
\begin{document}
\maketitle
\begin{abstract}
    We study the joint scheduling of behind-the-meter distributed energy resources (DERs), including flexible loads, renewable generation, and battery energy storage systems, under net energy metering tariffs with demand charges. The problem is formulated as a stochastic dynamic program aimed at maximizing expected operational surplus while accounting for renewable generation uncertainty. We analytically characterize the optimal control policy and show that it admits a threshold-based structure. However, due to the strong temporal coupling of the storage and demand charge constraints, the number of conditional branches in the policy scales combinatorially with the scheduling horizon, as it requires a look-ahead over future states. To overcome the high computational complexity in the general formulation, an efficient approximation algorithm is proposed, which searches for the peak demand under a mildly relaxed problem. We show that the algorithm scales linearly with the scheduling horizon. Extensive simulations using two open-source datasets validate the proposed algorithm and compare its performance against different DER control strategies, including a reinforcement learning-based one. Under varying storage and tariff parameters, the results show that the proposed algorithm outperforms various benchmarks in achieving a relatively small solution gap compared to a theoretical upper bound.
\end{abstract}
\begin{IEEEkeywords}
Battery storage systems, demand charges, distributed energy
resources, dynamic programming, energy management,
flexible demands, net metering, peak searching.
\end{IEEEkeywords}
\section{Background and Literature}
\IEEEPARstart{T}{he} transition to a clean, resilient, and sustainable energy system hinges on the optimal coordination of behind-the-meter (BTM) DERs, including distributed generation (DG), flexible demands, and battery energy storage system (BESS) \cite{akorede_distributed_2010}. The operation and coordination of these resources are driven by both economic incentives and technical constraints. Under net energy metering (NEM), {\em prosumers} can bi-directionally transact energy and money with the distribution utility (DU), who levy a {\em buy} (import) rate or a {\em sell} (export) rate based on the net-consumption \cite{Alahmed&Tong:22IEEETSG}. 

With the increasing penetration of new and large electric loads such as electric vehicles, electric ovens, and heat pumps, DUs are increasingly supplementing volumetric charges\footnote{Charges based on the total net energy consumption.} with {\em demand charges}. These are defined as \$/kW fees imposed on the customer's peak demand recorded over a billing cycle, typically measured in 15-minute or hourly intervals, and are introduced to encourage smoother load profiles and reduce the operational costs of grid maintenance \cite{jin_optimal_2017}. In the parlance of DU tariff design, this mechanism is referred to as a {\em non-coincident peak demand charge}, which levies charges based on each customer’s individual peak demand over a billing period, regardless of when it occurs. This contrasts with a {\em coincident peak demand charge}, which charges customers based on their demand contribution during the system-wide peak. Non-coincident demand charges are more commonly employed by DUs due to their administrative simplicity and alignment with cost recovery objectives for distribution infrastructure \cite{Stern:13NRELPeakDemand}.

Demand charges are a common component of commercial and industrial electricity tariffs and are increasingly being extended to residential customers in some jurisdictions.\footnote{Demand charges differ from {\em critical peak pricing}, a form of time-varying pricing aimed at incentivizing load shedding during critical periods \cite{Neufeld:87JEH}.} 

Utilities such as Pacific Gas and Electric (PG\&E), Consolidated Edison (Con Edison), Florida Power \& Light (FPL), and Arizona Public Service (APS) apply demand charges that vary based on customer class, time-of-use periods, and peak-coincident intervals \cite{PGEDemandCharge, ConEdisonTariff, FPLTariff,APSTariff}. These charges can constitute a substantial portion of monthly electricity bills, thereby incentivizing customers to strategically manage their load profiles. In response, utilities are adopting demand-based pricing not only to better reflect the cost of capacity provision but also to encourage peak load reduction and grid stability. From a system operator’s perspective, demand charges serve as a proxy for managing distribution-level constraints and capacity costs. Consequently, optimizing DER operations under such pricing schemes is essential for both economic efficiency and system reliability.

With the growing adoption of demand charges, several studies have highlighted the role of BTM DER in helping customers mitigate them \cite{StorageDemandCharge:17NREL, jin_optimal_2017, Luo&King&Ranzi&Dong:20TSG}.

This work focuses on the joint scheduling of BTM DER, including flexible demand, renewable DG, and BESS, under NEM frameworks with differentiated buy and sell rates and demand charges (Fig. \ref{fig:der_setup}). A comparative overview of related studies in terms of demand flexibility, storage integration, NEM capabilities, and demand charges is provided in Table \ref{tab:comparison}. To the best of our knowledge, this work is the first to co-optimize flexible demand and storage operations under the combined presence of NEM and demand charges.



\subsection{Related Work}

Previous research has explored various angles of home energy management systems (HEMS), particularly in optimizing prosumer energy costs and integrating renewable energy sources. Studies such as \cite{Harshah&Dahleh:15TPS,Zhang&Leibowicz&Hanasusanto:20TSG} utilized dynamic programming to minimize prosumer costs by scheduling BESS operations under inelastic demand, focusing on storage optimization without considering demand flexibility or NEM. Similarly, \cite{Chen&Wang&Heo&Kishore:13TSG,Khezri&Mahmoudi&Haque:21TSE} investigated the co-optimization of demand flexibility and BESS, addressing flexible demands and storage but omitting NEM's bidirectional energy transactions, which complicates the decision problem.

The integration of NEM, which allows prosumers to buy and sell energy, has been addressed in works like \cite{Li&Dong:19TSG,Xu&Tong:17TAC,alahmed_co-optimizing_2024,Jeddi&Mishra&Ledwich:20TSE}. For instance, \cite{Li&Dong:19TSG} optimized storage scheduling under NEM, while \cite{Xu&Tong:17TAC,alahmed_co-optimizing_2024,Jeddi&Mishra&Ledwich:20TSE} jointly optimized storage and flexible demands. However, these studies often neglect demand charges, which significantly complicate the scheduling problem due to their impact on peak load costs.

In contrast, while \cite{jin_optimal_2017,Luo&King&Ranzi&Dong:20TSG} incorporated demand charges into their optimization models, addressing peak power consumption penalties, the assumed free disposal of surplus renewables eliminates NEM's bidirectional energy transactions, which is a critical limitation given the growing real-world prevalence of such interactions. To provide a broader context, Zainab et al. \cite{Zainab&Ali&Ahmed&Syed&Zia:18Access} conducted a comprehensive review of HEMS, encompassing demand response programs, smart technologies, and intelligent controllers.

Recent progress in reinforcement learning (RL) has further enriched the area by enabling data-driven approaches to jointly optimize DER under complex tariffs. For example, Lee et al. \cite{Lee&Park&Sim&Lee:23TSG} introduced a federated RL framework for multi-residential energy scheduling that incorporates time-of-use and demand charge tariffs alongside energy storage and renewable sources. Similarly, \cite{Li&Dong:23AE} developed a decision framework for demand-side management under NEM, aiming at peak load reduction and improved renewable utilization. While RL offers scalability and adaptability to uncertain environments, it generally lacks the interpretability and optimality guarantees of dynamic programming, which can make it difficult to validate solutions or ensure consistent performance across varying operating conditions.

Despite these advancements, a gap remains in the literature for HEMS that co-optimize flexible demands and storage decisions under NEM and demand charges. Although \cite{jin_optimal_2017,Luo&King&Ranzi&Dong:20TSG} addressed demand charges and \cite{Li&Dong:23AE} emphasized NEM and demand management, few studies combine all four elements. This work addresses this gap by proposing a HEMS that maximizes prosumer surplus under NEM, incorporating demand flexibility, storage scheduling, and demand charges, as summarized in Table~\ref{tab:comparison}, and explained in the next section.

\begin{table}[ht]
    \centering
    \renewcommand{\arraystretch}{1.3}
    \caption{Sample of Related Work on Home Energy Management}
    \label{tab:comparison}
    \resizebox{\columnwidth}{!}{
    \begin{tabular}{@{}>{\centering\arraybackslash}m{2cm}cccc@{}}
        \toprule
        \textbf{Related Work} & \textbf{Flexible Demands} & \textbf{Storage} & \textbf{NEM (Buy/Sell)} & \textbf{Demand Charges} \\ 
        \midrule
        \cite{Harshah&Dahleh:15TPS,Zhang&Leibowicz&Hanasusanto:20TSG} & & \checkmark & & \\
        \cite{Chen&Wang&Heo&Kishore:13TSG,Khezri&Mahmoudi&Haque:21TSE} & \checkmark & \checkmark & & \\ 
        \cite{Li&Dong:19TSG} & & \checkmark & \checkmark & \\ 
        \cite{jin_optimal_2017} & & \checkmark & & \checkmark \\ 
        \cite{Luo&King&Ranzi&Dong:20TSG} & \checkmark & \checkmark & & \checkmark \\ 
        \cite{Xu&Tong:17TAC,alahmed_co-optimizing_2024,Jeddi&Mishra&Ledwich:20TSE} & \checkmark & \checkmark & \checkmark & \\ 
        \cite{Lee&Park&Sim&Lee:23TSG} & \checkmark & \checkmark & & \checkmark \\ 
        \cite{Li&Dong:23AE} & \checkmark & & \checkmark & \\ 
        \textbf{This Work} & $\pmb{\checkmark}$ & $\pmb{\checkmark}$ & $\pmb{\checkmark}$ & $\pmb{\checkmark}$ \\ 
        \bottomrule
    \end{tabular}
    }
\end{table}

\subsection{Summary of Results and Contributions} 
 This paper advances research on optimal scheduling of BTM DERs with four key contributions:
\begin{itemize}[leftmargin=*]
    \item We develop a stochastic dynamic programming framework that jointly optimizes flexible demand and BESS scheduling under NEM and demand charge constraints, maximizing prosumer operational surplus.
    \item We show that the optimal policy maintains a threshold structure consistent with previous work \cite{jin_optimal_2017} while demonstrating how the integration of multiple factors, such as demand flexibility, battery storage, and demand charges under bi-directional power flow, leads to prohibitive computational costs of the optimal policy.
    \item To overcome the prohibitive computational costs of the optimal solution, we propose an efficient algorithm that optimally searches for the prosumer's peak demand under a mildly relaxed problem by leveraging key derived properties from the dynamic programming formulation.
    \item Through extensive simulations against baseline approaches and a theoretical maximum, we demonstrate how a computationally efficient algorithm outperforms the benchmarks, including a proximal policy optimization-based reinforcement learning algorithm.
\end{itemize}

The works in \cite{jin_optimal_2017, alahmed_co-optimizing_2024} are most closely related to ours, as both adopt a stochastic dynamic programming framework for the prosumer decision problem. Our approach departs from \cite{jin_optimal_2017} in two key aspects: (i) we model demand as flexible and jointly optimize it with battery storage, and (ii) we allow for bi-directional power and monetary flows between the prosumer and the DU under asymmetric import and export prices. Symmetric NEM import and export prices, {\em i.e.,} NEM 1.0, are known to create utility revenue shortfalls \cite{NEMevolution:23NAS} potentially leading to death spirals \cite{DARGHOUTH2016713} and cross-subsidies \cite{alahmed_integrating_2022}. Consequently, several states, including California, Arizona, Illinois, Georgia, and Utah, have adopted successor NEM tariffs featuring asymmetric import and export prices \cite{DSIRE}.

Compared to \cite{alahmed_co-optimizing_2024}, our model incorporates demand charges in addition to the volumetric charge, introducing stronger temporal coupling across decision stages. From an economic perspective, demand charges improve the cost reflectivity of the tariff and reduce the over-reliance on volumetric charges to recover bundled cost components \cite{NEMevolution:23NAS}.

Lastly, unlike both \cite{jin_optimal_2017, alahmed_co-optimizing_2024}, our work proposes an efficient approximation algorithm that outperforms various benchmarks in achieving a relatively small
solution gap compared to a theoretical upper bound.

\subsection{Mathematical Notations and Paper Organization}

Table \ref{tab:MajorSymbols} lists the major symbols used. Throughout the paper, boldface letters indicate column vectors as in $\bm{x}=(x_1,\ldots, x_n)$ with $\bm{x}^\top$ as its transpose. In particular, $\bm{1}$ is a column vector of all ones.  For a multivariate differentiable function $f$ of $\bm{x}$, we use interchangeably $f(\bm{x})$ and $f(x_1,\ldots,x_n)$, and we use $f'$ and $f''$ to denote its first-order and second-order derivative (if it exists), respectively. We use $f^{-1}$ to denote its inverse function (if it exists). For vectors $\bm{x},\bm{y}$,  $\bm{x} \preceq \bm{y}$ is the element-wise inequality $x_i \le y_i$ for all $i$, and $[\bm{x}]^+, [\bm{x}]^-$ are the element-wise positive and negative parts of vector $\bm{x}$, i.e., $[x_i]^+:=\max\{0,x_i\}$, $[x_i]^- :=-\min\{0,x_i\}$ for all $i$, and $\bm{x}= [\bm{x}]^+ - [\bm{x}]^-$. We denote by $\mathbb{R}_+:=\{ x\in \mathbb{R}: x \geq 0\}$ the set of non-negative real numbers.

The rest of the paper is organized as follows. \sectsign\ref{sec:problem}, provides a mathematical formulation of DERs, tariff under demand charge, system dynamics, and the joint scheduling decision problem. \sectsign \ref{sec:method} presents some analytical results on the dynamic programming problem, and \sectsign\ref{sec:specialcases} proposes an efficient approximation algorithm that circumvents the exponential complexity of solving the dynamic program. A detailed case study using real data from two buildings is presented in \sectsign\ref{sec:case_study}, followed by concluding remarks and future directions in \sectsign\ref{sec:conclusion}.

\begin{table}[ht]
\centering
\caption{Major variables and parameters (alphabetically ordered)}
\label{tab:MajorSymbols}
\vspace{-0.2cm}
\resizebox{\columnwidth}{!}{%
\begin{tabular}{@{}ll@{}}
\midrule \midrule
Symbol                                                   & Description                                                            \\ \midrule
$B$ &  Battery storage capacity.\\
$c_t$ & Prosumer’s peak demand before time $t$.\\
$\bm{d}_t$                                      & Vector of consumption bundle at time $t$.         \\
$d_t^k$                                      & Consumption of device $k$ at time $t$.         \\
$\overline{\bm{d}}$              & Consumption bundle's upper limit.             \\
$e_t$           & Battery storage output at time $t$\\
$\underline{e},\overline{e}$  & Battery storage discharging and charging limits.\\
$g_t$                                           & Renewable generation at time $t$.           \\
$\gamma$ & Salvage value rate.\\                               
$k$                                                 & Index of consumption devices.                                  \\
$P_t(\cdot)$             & Payment function at time $t$.                    \\
$p$ & Peak demand price.\\
$p^+, p^-$                                           & NEM X buy (retail) and sell (export) rates.                            \\
$r_t(\cdot)$ & Reward function at time $t$.\\
$\rho$ & Battery storage discharging efficiency.\\
$s_t$               &       Storage state of charge at time $t$.\\
$t$                                                  & Index of time.  \\   
$\tau$ & Battery storage charging efficiency.\\
$u_t$ & System control action at time $t$.\\
$U_t(\cdot)$                                             & Utility of consumption at time $t$.                                  \\
$v_t$ & Net consumption at time $t$ excluding renewables.\\
$V^\pi, V^\ast$ & Value function under policy $\pi$ and optimal value function.\\
$x_t$ & System state.\\
$z_t$                                           & Net consumption at time $t$. \\
\midrule \midrule
\end{tabular}%
}
\end{table}

\section{Problem Formulation}
\label{sec:problem}
\begin{figure}[t]
    \centering
    \includegraphics[width=0.9\linewidth]{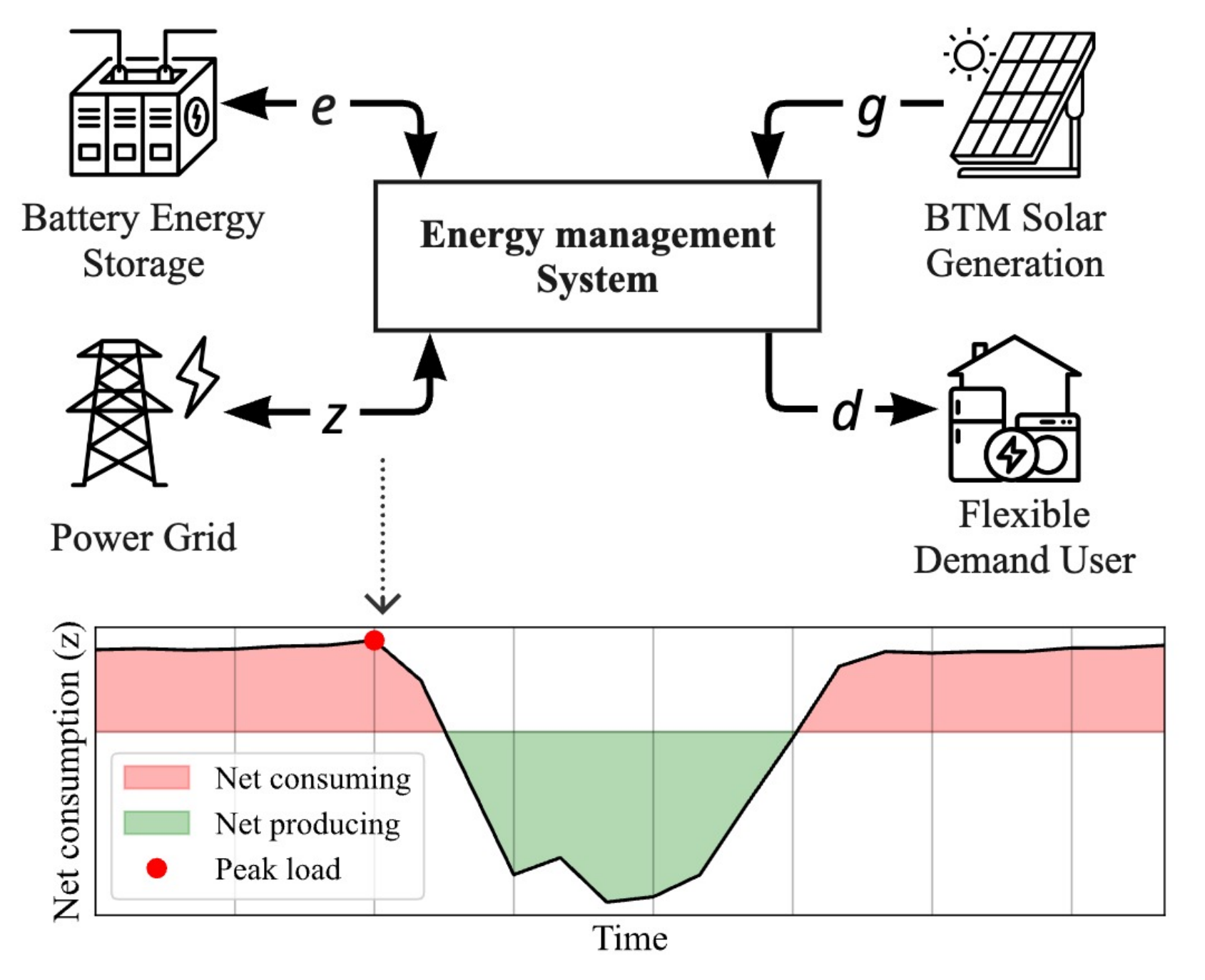}
    \vspace{-0.2cm}
    \caption{DER setup with flexible demand $d \in \mathbb{R}+$, renewable generation $g \in \mathbb{R}+$, storage operation $e \in \mathbb{R}$, and net consumption $z \in \mathbb{R}$. The illustrated NEM tariff includes both energy and demand charges, with costs associated with net consumption, net generation, and the peak load recorded over the daily billing period.}
    \label{fig:der_setup}
\end{figure}

We study a HEMS for prosumers that co-optimizes flexible demand, storage operation, and renewable distributed generation under a NEM tariff with distinct buy/sell rates and peak demand charges. The optimization objective maximizes the operational surplus defined as consumption utility minus electricity costs, where demand charges are levied on the maximum consumption drawn from either load, battery charging, or both. As shown in Fig.~\ref{fig:der_setup}, the energy storage and flexible loads are jointly scheduled to reduce the demand charge by shaving the load curve, and reduce the volumetric charge by arbitraging NEM rates, while still satisfying the prosumer's consumption.

\subsection{Prosumer Resources}
We consider the operation over a discrete and finite time horizon $t=0,1,\ldots, T$. As shown in Fig.\ref{fig:der_setup}, the prosumer has {\em BTM renewable DG}, denoted by $g_t$. The renewable DG is modeled as an exogenous (positive) Markov random process.

The prosumer {\em flexible demands} are denoted by the vector 
\begin{equation}\boldsymbol{d}_t=(d_t^1, \ldots, d_t^K)\in\mathcal{D}\coloneqq\{\boldsymbol{d}:\boldsymbol{0}\preceq \boldsymbol{d}\preceq\overline{\boldsymbol{d}}\}\subseteq\mathbb R^K_+,
\end{equation}
where $K$ is the total number of devices, and $\overline{\boldsymbol{d}}$ is the consumption bundle’s upper limit.\footnote{The minimum consumption limit is set to $\bm{0}$ without loss of generality.} Some important devices might be less flexible or {\em inflexible}. Mathematically, this is incorporated by assigning these devices stricter flexibility limits, {\em i.e.}, for an inflexible device $k$, whose consumption should not deviate from $\Tilde{d}_t^k$, its flexibility limits would be $\Tilde{d}_t^k \leq d_t^k \leq \Tilde{d}_t^k$.



\added{The prosumer also owns a BESS, which has output $e_t\in[-\underline{e}, \overline{e}]$, where $\underline{e}$ and $\overline{e}$ are the discharging and charging limits, respectively}. The output is composed of charging $[e_t]^+=\max\{e_t, 0\}$ and discharging $[e_t]^-=\max\{-e_t, 0\}$ actions. Given charging and discharging efficiencies, $\tau \in(0, 1]$ and $\rho\in(0, 1]$, respectively, the BESS state of charge (SoC) $s_t\in[0, B]$ evolves as: 
\begin{equation}\label{eq:SoC}
    s_{t+1} = s_t + \tau [e_t]^+-[e_t]^-/\rho,\quad t=0,\ldots, T-1,
\end{equation} 
where $B$ denotes battery storage capacity.

The {\em net energy consumption} of the prosumer at time $t$ is
\begin{equation}
    z_t := \bm{1}^\top\boldsymbol{d}_t + e_t - g_t.
\end{equation} 
Therefore, the prosumer imports power from the grid if $z_t>0$ and exports power to the grid if $z_t<0$.

\subsection{System Dynamics}
The {\em system state} at time $t$, $x_t$, is described by $x_t = [s_t, g_t, c_t]$,
where $c_t$ is the prosumer's peak demand before time $t$. The system is controlled by the BESS action, i.e., charging or discharging, and consumption adjustments, hence the control action is $u_t \coloneqq (e_t, \boldsymbol{d}_t)$. The peak demand in the state evolves as
\begin{equation}\label{eq:DCevolution}
    c_{t+1} = \max_{0\le i\le t}z_i=\max\{z_t,c_t\}, \quad t=0,1,\ldots, T-1,
\end{equation}
and the battery SoC evolves as per (\ref{eq:SoC}).

\subsection{NEM Tariff with Demand Charges}
Under NEM, when the prosumer's gross load (from consumption and storage charging) is higher than its generation (from renewable DG and storage discharging), the prosumer pays at the retail rate. Otherwise, the prosumer gets compensated at a predetermined export rate. 
To better reflect the capacity costs incurred by the DU, a demand charge is added to the prosumer payment. The demand charge is a one-time payment based on the peak demand over the billing period.

Let $(p_t^+\ge p_t^-, p) \ge 0$ be the retail, export, and peak demand prices, respectively. The payment under NEM $P_t(\cdot)$ over the considered time horizon is given by 
\begin{equation}\label{eq:paymentAgg}
  \sum_{t=0}^{T-1}  P_t(z_t, c_t) = \sum_{t=0}^{T-1} \left( p_t^+[z_t]^+ - p_t^-[z_t]^-\right)+ p c_T +A,
\end{equation}
where $A$ denotes a fixed charge, commonly known as a {\em connection charge}, that does not vary with consumption \cite{Alahmed&Tong:22IEEETSG}. The payment in (\ref{eq:paymentAgg}) corresponds to a {\em three-part tariff} structure, consisting of volumetric, demand, and fixed components \cite{Nieto:16EJ}.

For the convenience of analysis, we equivalently break down the peak demand charge at the end of time $T$ to each time step $t$ as a rolling-based incremental charge, since
$$p c_T = p\sum_{t=0}^{T-1} \left(c_{t+1}-c_t\right)\overset{(\ref{eq:DCevolution})}{=} p\sum_{t=0}^{T-1}[z_t-c_t]^+.$$ Thus, we can write the payment under NEM X in (\ref{eq:paymentAgg}) for each time $t$ as 
\begin{equation}\label{eq:payment}
    P_t(z_t, c_t) = p_t^+[z_t]^+ - p_t^-[z_t]^- + p[z_t-c_t]^+ + A_t,
\end{equation}
where, for every time $t$, $P_t(\cdot)>0$ ($P_t(\cdot)<0$) indicates that the prosumer pays (gets compensated) to (by) the DU.\footnote{We assume $A_t=0, \forall t$, since fixed costs do not influence DER scheduling decisions, at least in the short-run.}

Special cases of the general NEM tariff in (\ref{eq:payment}) are considered in \cite{alahmed_co-optimizing_2024} when demand charges ($p[z_t-c_t]^+$) are absent and in \cite{jin_optimal_2017} when energy exports are absent ($p_t^-[z_t]^-$).

\subsection{Prosumer Utility Function}
We employ utility functions to formally represent prosumer preferences over consumption bundles, following the standard framework in microeconomic theory \cite{MasColell&Whinston:95BookOxford}. At each time step $t$, the prosumer's utility of consuming $\boldsymbol{d}_t$ is given by an additive utility function
\begin{equation}\label{eq:utility}
    U_t(\boldsymbol{d}_t) = \sum_{i=1}^K U_t(d_t^i), \quad t=0,1,\ldots, T-1
\end{equation}
that is assumed to be concave and continuously differentiable. Without loss of generality, we assume $\nabla U_t(\boldsymbol{d}_t)\ge 0$, $\boldsymbol{d}\in\mathcal D$.

\subsection{DER Co-Optimization Problem}\label{subsec:DERcooptimization}
To optimally schedule its DER, the prosumer solves a dynamic program that maximizes its reward function by explicitly accounting for the temporal coupling of storage operations and a NEM tariff with asymmetric retail and export rates and peak demand charges.

The reward function is defined as the difference between the utility of consumption (\ref{eq:utility}) and payment under NEM and demand charges (\ref{eq:payment}). Therefore, the reward function is

\begin{equation}
\begin{aligned}
    r_t(x_t, u_t) &= r_t((s_t, g_t, c_t), (e_t, \mathbf{d}_t))\\ &= U_t(\boldsymbol{d}_t) - P_t(z_t, c_t),
    \quad t=0,\ldots, T-1.
\end{aligned}
\end{equation}
To encourage the preservation of energy in the battery, we define a linear reward to the remaining energy at the terminal stage $T$, {\em i.e. }
\begin{equation}
    r_T(x_T)=r_T(s_T, g_T, c_T)=\gamma s_T,
\end{equation}
where $\gamma>0$ is the value assigned to the remaining charge in the storage at the end of the horizon.



The policy $\bm{\pi}=(\pi_0, \ldots, \pi_{T-1})$ is a sequence of mappings from state space to action space such that $\pi_t(x_t) = u_t$, for every $t$. Given initial state $x_0=(s, g, 0)$, the storage-consumption co-optimization is then defined as:
\begin{subequations}\label{eq: prob}
\begin{align}
    \max_{\pi}\quad& \mathbb E\left[\sum_{t=0}^{T}r_t(x_t, u_t)\right]\\
    \text{Subject to}&\quad \text{for all}\ t=0, 1,\ldots, T-1\nonumber\\
    & x_0 = (s, g, 0)\\
    &u_t=(e_t, \boldsymbol{d}_t)=\pi_t(x_t)\\
    & e_t\in[-\underline{e}, \overline{e}] \\
    & s_{t+1} = s_t + \tau[e_t]^+-[e_t]^-/\rho \label{SoCevolution}\\
    & 0 \leq s_t \leq B \label{eq:StorageCapacity}\\
    & g_{t+1}\sim F_{\cdot |g_t} \\
    & c_{t+1}=\max\{\boldsymbol{1}^\top\boldsymbol{d}_t+e_t - g_t, c_t\} \label{DCevolution}\\
    & \boldsymbol{0} \preceq \boldsymbol{d}_t\preceq \boldsymbol{\overline{d}}.
\end{align}
\end{subequations}
Constraints \eqref{SoCevolution} and \eqref{DCevolution} induce intertemporal coupling, precluding a myopic (stage-wise) optimal policy. Specifically, relaxing either constraint in isolation is insufficient to eliminate the temporal dependencies, and the resulting optimal policy remains non-myopic (see Theorem \ref{thm: myopic} in Appendix \ref{app:theory}). However, when both constraints are simultaneously relaxed, the problem admits a myopic optimal policy characterized by decoupled per-period decisions (see Lemma~\ref{lem:OptimalRelaxed}).

The value function $V_t^\pi(\cdot)$ is defined as the expected total reward received under the policy $\pi$, i.e., $V_t^\pi(x) = \mathbb E[\sum_{t=0}^T r_t(x_t, \pi(x_t))|x_t=x]$. The optimal value function is defined as $V_t^*(x)=\max_\pi V_t^\pi(x)$.

\section{Dynamic Program Analytical Results}
\label{sec:method}
In this section, we \replaced{show that the optimal policy for the co-optimization problem in \sectsign\ref{subsec:DERcooptimization} is a threshold-based one.}{characterize an optimal threshold policy for the dynamic program formulated in \sectsign\ref{subsec:DERcooptimization}.}

First, we show that the optimal value function is concave.
\begin{lemma}\label{lem: concave}
    The optimal value function $V^*_t(s_t, g_t, c_t)$ is concave in $(s_t, c_t)$ for every time step $t\le T$. 
\end{lemma}
\begin{proof}
We follow the proof technique for Lemma 1 in \cite{jin_optimal_2017}. 

\replaced{Consider two different states $x_\tau$ and $x_\tau'$ with their corresponding optimal trajectories $\{(x_t, u_t)\}_{t=\tau}^T$ and $\{(x_t', u_t')\}_{t=\tau}^T$.}{Let $x_\tau^{(0)}$ and $x_\tau^{(1)}$ be two different states and $\{(x_t^{(0)}, u_t^{(0)})\}_{t=\tau}^T$ and $\{(x_t^{(1)}, u_t^{(1)})\}_{t=\tau}^T$ be their corresponding optimal trajectory.} \replaced{Let $\tilde x_\tau=(x_\tau+x_\tau')/2$ be their average initial state, and define the average policy $\tilde \pi=\{\tilde u_t=(u_t+u_t')/2\}_{t=\tau}^T$ with corresponding trajectory $\{\tilde x_t\}_{t=\tau}^T$.}{We consider the average initial state $\tilde x_\tau=(x_\tau^{(0)}+x_\tau^{(1)})/2$ and the trajectory $\{\tilde x_t\}_{t=\tau}^T$ under average policy $\tilde \pi=\{\tilde u_t=(u_t^{(0)}+u_t^{(1)})/2\}_{t=\tau}^T$.}

\replaced{First, by the concavity of $U_t$:}{By definition, $U_t$ is concave:}
\begin{equation}
    (U_t(\bm{d}_t) + U_t(\bm{d}_t'))/2\le U_t((\bm{d}_t+\bm{d}_t')/{2})=U_t(\tilde{\bm{d}}_t).\label{eq: u}
\end{equation}
\replaced{The cost $P_t(z_t, c_t)$ consists of two components: the energy cost $p_t^+[z_t]^+-p_t^-[z_t]^-$ and the demand charging $p[z_t-c_t]^+$. Since $p_t^+\ge p_t^-\ge0$, the energy cost is convex.}{We separate the $P_t(z_t, c_t)$ into energy cost $p_t^+[z_t]^+-p_t^-[z_t]^-$ and the demand charging cost $p[z_t-c_t]^+$. By assumption, $p_t^+\ge p_t^-\ge 0$, and $p_t^+[0]^+=p_t^-[0]^-=0$, the left derivative at $0$ is smaller than the right derivative at $0$. So, the derivative of $p_t^+[z_t]^+ - p_t^-[z_t]^-$ is non-decreasing and hence $p_t^+[z_t]^+ - p_t^-[z_t]^-$ is convex.} Then, 
\begin{equation}
\begin{aligned}
    &-\sum_{t=\tau}^T(p_t^+[z_t]^+-p_t^-[z_t]^-+p_t^+[z_t']^+-p_t^-[z_t']^-)/ 2\\
    &\le -\sum_{t=\tau}^T p_t^+[(z_t + z_t')/2]^+-p_t^-[(z_t+z_t')/2]^-\\
    &=-\sum_{t=\tau}^Tp_t^+[\tilde z_t]^+-p_t^-[\tilde z_t]^-,\label{eq: energy}
\end{aligned}
\end{equation}
and the demand charging cost becomes
\begin{equation}
    \begin{aligned}
        &-\sum_{t=\tau}^T(p[z_t-c_t]^++ p[z_t'-c_t']^+)/2 \\
        =&-p(\max_{\tau\le t\le T}\{z_t\}-c_\tau+\max_{\tau\le t\le T}\{z_t'\}-c_\tau')/2\\
        \le&-p(\max_{\tau\le t\le T}\{(z_t + z_t')/2\} - (c_\tau+c_\tau')/2)\\
        =&-p(\max_{\tau\le t\le T}\{\tilde z_t\} - \tilde c_\tau)=\sum_{t=\tau}^Tp[\tilde z_t - \tilde c_t]^+.
    \end{aligned}\label{eq: charge}
\end{equation}
Combining (\ref{eq: u})-(\ref{eq: charge}), 
\begin{equation}
    V^*_\tau(\tilde x_\tau)\ge V_\tau^{\tilde\pi}(\tilde x_\tau)\ge (V_\tau^*(x_\tau) + V_\tau^*(x_\tau'))/2,
\end{equation}
which proves the concavity of the optimal value function. 
\end{proof}

\added{It turns out that the optimal value function is also non-decreasing in the system state, which is formally stated in the next lemma.}
\begin{lemma}\label{lem:monotonicity}
    The optimal value function $V_t^*(s_t,g_t,c_t)$ is non-decreasing in $s_t$\added{, $g_t$,} and $c_t$ for every time step $t\le T$.
\end{lemma}
\begin{proof}
    (1) $s_t$: Consider two SoCs $s_t$, $s_t'$ such that $s_t=s_t'+\epsilon$, $0<\epsilon\ll \underline{e}$. Let $\{u_\tau'=(e_\tau', \bm{d}_\tau')\}_{\tau=t}^T$ be the optimal action sequence for SoC $s_t'$. If there exists $t\le \tau_0< T$, $e_{\tau_0}'\ge -\underline{e}+\epsilon$, action sequence $\{u_\tau=(e_\tau'-\epsilon\cdot\mathbf 1[\tau=\tau_0], \bm{d}_\tau')\}_{\tau=t}^T$ is feasible for SoC $s_t$. Hence $V^*_t(s_t, g_t, c_t)-V^*_t(s_t',g_t,c_t)\ge P_{\tau_0}(z_{\tau_0}', c_{\tau_0}')-P_{\tau_0}(z_{\tau_0}'-\epsilon, c_{\tau_0}')\ge 0$. Otherwise, $e_\tau'<-\underline{e}+\epsilon<0$, $t\le\tau\le T$, so action sequence $\{u_\tau'\}_{\tau=t}^T$ is feasible for SoC $s_t$, $V^*_t(s_t, g_t, c_t)-V^*_t(s_t',g_t,c_t)\ge \gamma(s_T-s_T')=\gamma\epsilon>0$.    
    
    (2) $g_t$: Consider two generations $g_t=g_t'+\epsilon$, $\epsilon>0$. Let $\pi'=\{u_\tau'=(e_\tau', \bm{d}_\tau')\}_{\tau=t}^T$ be the optimal action sequence for generation $g_t'$. Apply the action sequence to generation $g_t$, $z_{t+1}=e_t'+\bm{1}^\top\bm{d}_t'-g_t=z_{t+1}'-\epsilon$, $V_t^*(s_t, g_t, c_t)-V_t^*(s_t, g_t', c_t)\ge P_t(z_t', c_t) - P_t(z_t'-\epsilon, c_t)\ge 0$.     
    
    (3) $c_t$: From (\ref{eq:payment}), $P_t(z_t,c_t)$ is non-increasing in $c_t$, hence $V_t^*(s_t,g_t, c_t)$ is non-decreasing in $c_t$.
\end{proof}

\added{Lastly, we derive a property that links the optimal $Q$-function, defined as $Q_t^*(x_t, u_t) = r(x_t, u_t) + \mathbb E[V_{t+1}^*(x_{t+1})]$, to the system's control action $u_t$.}
\begin{lemma}\label{lem:Qfunction}
The optimal $Q$-function $Q_t^*(x_t, u_t)$ is concave in $u_t$.     
\end{lemma}
\begin{proof}
For $u_\tau\neq u_\tau'$, let $\{(x_t, u_t)\}_{t=\tau}^T$ be the trajectory of $Q_\tau^*(x_\tau, u_\tau)$ and $\{(x_t', u_t')\}_{t=\tau}^T$ be the trajectory of $Q_\tau^*(x_\tau, u_\tau')$. By (\ref{eq: u}) - (\ref{eq: charge}), applying actions $\{(u_t+u_t')/2\}_{t=\tau}^T$ to $x_\tau$ gives:
\begin{equation*}
    \begin{aligned}
        Q_\tau^*(x_\tau, u_\tau)+Q_\tau^*(x_\tau, u_\tau')
        \le&\mathbb E\left[\sum_{t=\tau}^Tr_t(x_t, (u_t+u_t')/2)|x_\tau\right]\\
        \le& Q_\tau^*(x_\tau, (u_\tau+u_\tau')/2).
    \end{aligned}
\end{equation*}
    
Therefore, $Q_t^*(x_t, u_t)$ is concave in $u_t$.
\end{proof}

Building on the results of Lemma~\ref{lem: concave}, Lemma~\ref{lem:monotonicity}, and Lemma~\ref{lem:Qfunction}, we establish Theorem~\ref{thm: opt}, which describes the structure of the optimal co-optimization policy under demand charges.
\begin{theorem}\label{thm: opt}
    For each step $t=0, 1, \ldots, T-1$, the optimal policy is a threshold policy. 
\end{theorem}
\begin{proof}
    The feasible region for action $u_t$, $(e_t, \bm{d}_t)\in \mathcal F_t= [\max\{-\underline{e}, -\rho s_t\}, \min\{\overline{e}, (B-s_t)/\tau\}]\times\prod_{i=1}^K[0, d_i]$ is a compact set\replaced{.}{ with $(K+1)2^K$ boundary segments.} From Lemma~\ref{lem: concave} to Lemma~\ref{lem:Qfunction}, the optimal value function and $Q$-function is concave in action $u_t$, so the optimal action
    \begin{equation}
        u_t^*=\arg\max_{u\in\mathcal F_t} \{r_t(x_t, u) + Q_t^*(x_t, u)\}
    \end{equation}
    \replaced{must exist in the interior of the feasible set $\mathcal F_t$ or at its boundary. The optimal action can be obtained by solving concave optimizations with and without boundary constraints, which results in a threshold policy.}{ So, we can get the optimal solution by solving $(K+1)2^K$ constrained concave optimization and one unconstrained concave optimization, which results in a threshold solution of $(K+1)2^K+1$ cases.}
\end{proof}
Although Theorem~\ref{thm: opt} reveals a structured policy, direct dynamic programming implementation remains impractical. If we directly apply concave optimization on the compact sets, it necessitates gradients of $r_t$ and $Q_t$, which are computationally prohibitive, especially for complex utility functions that may not be additive with respect to consuming devices. On the other hand, if we want to avoid the expensive optimization by backtracking, the continuous-space $Q$-function tracking requires costly discretization, as validated in simpler settings in~\cite{jin_optimal_2017}. To this end, we develop, in \sectsign\ref{sec:specialcases}, an efficient approximation algorithm to overcome these computational barriers.

\section{An Efficient Approximation Algorithm}
\label{sec:specialcases}

Since obtaining the exact optimal solution to problem (\ref{eq: prob}) is computationally intractable, we develop an efficient approximation algorithm. The proposed large storage peak searching (LSPS) algorithm efficiently searches for the optimal solution by solving a relaxed version of the original problem. The resulting solution is then projected onto the feasible region $\mathcal{F}_t$, producing near-optimal and computationally tractable approximations. The original formulation prevents myopic solutions through constraints (\ref{SoCevolution}) and (\ref{DCevolution}). To address this computational challenge, we first relax constraint (\ref{SoCevolution}) by considering a large storage capacity regime (see \sectsign\ref{sec: relax_11e}), then develop an efficient method to handle constraint (\ref{DCevolution}) (see \sectsign\ref{sec: relax_11h}).

\subsection{Relaxation of the DER Co-Optimization Problem}\label{sec: relax_11e}
We consider a large storage capacity regime as a relaxation of the battery capacity constraint (\ref{SoCevolution}). Observing that $e_t$ and $\bm{d}_t$ consistently appear together in the expression $\mathbf{1}^\top\bm{d}_t+e_t$, we introduce a helper function $h_t(v_t)$:
\begin{equation}\label{eq: helper}
    \begin{aligned}
        h_t(v_t) := \max\{&U_t(\bm{d}_t)+\gamma e_t: \bm{1}^\top\bm{d}_t+e_t=v_t, \\
        &\bm{0}\preceq \bm{d}_t\preceq \overline{\bm{d}}, -\underline{e}\le e_t\le \overline{e}\},
    \end{aligned}
\end{equation}
where $v_t$ represents the net consumption excluding renewables at time step $t$. 
The function $h_t(v_t)$ is differentiable, monotonically increasing, and concave for all time steps $t$. For any $v_t\in[-\underline{e}, \overline{e}+\bm{1}^\top\bm{\overline{d}}]$, there exists a valid pair $(e_t, \bm{d}_t)$ that achieves $h_t(v_t)$ with $\mathbf{1}^\top\bm{d}_t+e_t=v_t$. This equivalence allows us to optimize over $v_t$ after relaxing constraint (\ref{SoCevolution}).

Using the helper function (\ref{eq: helper}), we formulate the relaxed problem as
\begin{subequations}\label{prob: large}
    \begin{align}
        \max_{v_0,\ldots, v_{T-1}}\quad & \sum_{t=0}^{T-1} h_t(v_t)-p_t^+[v_t-g_t]^+ + p_t^-[v_t-g_t]^- \nonumber\\
        &- p[v_t-g_t-c_t]^+\\
        \text{Subject to}\quad & -\underline{e}\le v_t\le \overline{e} +\bm{1}^\top\overline{\bm{d}}\\
        &c_{t+1}=\max\{c_t, v_t-g_t\}.\label{eq: peak_demand_v}
    \end{align}
\end{subequations}
To eliminate the temporal dependencies introduced by constraint (\ref{DCevolution}), we consider a fixed peak demand bound $c$ in the large capacity regime (\ref{prob: large})
\begin{subequations}\label{prob: peak}
    \begin{align}
        \max_{v_0,\ldots, v_{T-1}}\quad & \sum_{t=0}^{T-1} h_t(v_t)-p_t^+[v_t-g_t]^+ + p_t^-[v_t-g_t]^--pc\\
        \text{Subject to}\quad & -\underline{e}\le v_t\le \overline{e} +\bm{1}^\top\overline{\bm{d}}, \label{eq: hard_bound}\\
        & v_t\le g_t+c. \label{eq: peak_bound}
    \end{align}
\end{subequations}
This reformulation replaces the recursive peak demand update (\ref{eq: peak_demand_v}) with a single-step upper-bound constraint (\ref{eq: peak_bound}), effectively removing temporal dependencies while preserving the optimal policy structure. We show that the optimal policy of (\ref{prob: large}) is preserved under this relaxation.
\begin{theorem}\label{thm: equal}
    Suppose $\{v_t^*\}$ and $c^*$ are optimal actions and peak demand for (\ref{prob: large}), then $\{v_t^*\}$ is also optimal for (\ref{prob: peak}) with demand bound $c=c^*$.
\end{theorem}
\begin{proof}
    Suppose $\{v_t^*\}$ is the optimal solution for problem (\ref{prob: large}). Since $c^*$ is the optimal peak demand, we have $v_t-g_t\le c^*=c$ for all $t$. This means $\{v_t\}$ is also feasible for problem (\ref{prob: peak}). Since problem (\ref{prob: peak}) searches within a subspace of problem (\ref{prob: large}), $\{v_t^*\}$ is also optimal for problem (\ref{prob: peak}).
\end{proof}
By Theorem~\ref{thm: equal}, we can compute the optimal policy for (\ref{prob: large}) by searching for the optimal $c^*$ and solving (\ref{prob: peak}).

\subsection{Solution for a Given Peak Demand}\label{sec: relax_11h}
With a fixed peak demand bound $c$, the optimization problem decomposes into independent single-step subproblems that can be solved analytically. Due to the concavity of $h_t$, we first derive the unconstrained optimum $v^\dagger_t(g_t)$, then project it onto the feasible region $\mathcal{F}_t$ defined by the constraints to obtain the constrained optimum $v^*_t(g_t)$. 
\begin{lemma}\label{lem:OptimalRelaxed}
    For $t=0,\ldots, T-1$, the optimal policy for (\ref{prob: peak}) is given by 
    \begin{align}\label{eq: myopic}
        &v_t^\ast(g_t)=\nonumber\\
        &\begin{cases}
             -\underline{e},  & \text{if}\quad  v^\dagger_t(g_t)< -\underline{e},\\
            v_t^\dagger(g_t), & \text{if}\quad -\underline{e}\le v^\dagger_t(g_t)\le \Gamma_t,\\
            \Gamma_t, &\text{if}\quad  \Gamma_t< v^\dagger_t(g_t),
        \end{cases}
    \end{align}
    where $\Gamma_t=\min\{c+g_t, \overline{e}+\bm{1}^\top\overline{\bm{d}}\}$ is the upper bound of $v_t^*(g_t)$, the unconstrained optimum $v_t^\dagger(g_t)$ is the solution without constraints (\ref{eq: hard_bound}) and (\ref{eq: peak_bound}), given by \begin{align}\label{eq: uc_v}
        &v_t^\dagger(g_t)\nonumber\\
        =&\begin{cases}
            (h')^{-1}(p_t^-), & \text{if}\quad (h')^{-1}(p_t^-)<g_t,\\
            g_t, & \text{if}\quad (h')^{-1}(p_t^+)\le g_t\le (h')^{-1}(p_t^-),\\
            (h')^{-1}(p_t^+), & \text{if}\quad g_t< (h')^{-1}(p_t^+).
        \end{cases}
    \end{align}
\end{lemma}
\begin{proof}
    With fixed peak demand, the optimization decomposes into independent subproblems at each time step $t=0,1,\ldots, T-1$:
    \begin{subequations}
        \begin{align}
            \max_{v_t}\quad & h_t(v_t) -p_t ^+[v_t-g_t]^++p_t^-[v_t-g_t]^-\\
            \text{Subject to}\quad & v_t\in[-\underline{e}, \overline{e}+\bm{1}^\top\overline{\bm{d}}], \quad v_t-g_t\le c. \label{eq: relaxed_constraint}
        \end{align}
    \end{subequations}
    The unconstrained optimum is   
    \begin{equation*}
    \begin{aligned}
        v_t^\dagger(g_t)=\arg\max_v h_t(v) - p_t^+[v-g_t]^+ + p_t^-[v-g_t]^-\\
        =
        \begin{cases}
            (h_t')^{-1}(p_t^-), & \text{if} \quad (h_t')^{-1}(p_t^-)<g_t,\\
            g_t, & \text{if} \quad (h_t')^{-1}(p_t^+)\le g_t\le (h_t')^{-1}(p_t^-),\\
            (h_t')^{-1}(p_t^+), & \text{if} \quad g_t< (h_t')^{-1}(p_t^+).
        \end{cases}
    \end{aligned}
    \end{equation*}
    Projecting onto the feasible region defined by constraint (\ref{eq: relaxed_constraint}) yields the optimal policy in (\ref{eq: myopic}).
\end{proof} 
The optimal policy has an intuitive economic interpretation based on the relationship between renewable generation and consumption needs. The unconstrained optimal solution $v_t^\dagger(g_t)$ balances the marginal utility of consumption against energy costs. When $(h_t')^{-1}(p_t^-) < g_t$, the renewable generation $g_t$ exceeds the energy demanded by consuming at the export price $p^-$ and charging the battery. In this case, the optimal consumption $v_t^\dagger(g_t) = (h_t')^{-1}(p_t^-)$ is set such that the marginal utility equals the export price, meaning excess generation beyond this consumption level should be sold to the grid. When $(h_t')^{-1}(p_t^+) \le g_t \le (h_t')^{-1}(p_t^-)$, the renewable generation exactly matches the optimal consumption level, so $v_t^\dagger(g_t) = g_t$ and no energy transactions occur. When $g_t < (h_t')^{-1}(p_t^+)$, renewable generation is insufficient relative to the consumption level that would justify purchasing energy at price $p_t^+$. Here, the optimal consumption $v_t^\dagger(g_t) = (h_t')^{-1}(p_t^+)$ sets marginal utility equal to the purchase price, requiring additional energy to be bought from the grid. The final policy $v_t^*(g_t)$ then projects this unconstrained optimum onto the feasible region, ensuring that consumption respects both battery operation limits and peak demand constraints.

\subsection{Algorithmic Search for the Optimal Peak Demand}
We leverage the concavity properties of the objective function to efficiently determine the optimal peak demand $c^*$ without solving the complete optimization problem (\ref{prob: large}). Our approach exploits the structural properties of the relaxed co-optimization problem (\ref{prob: peak}).
\begin{lemma}
    \label{lem:ConcavitySpecialCase}
    Let $J(c)$ be the optimal objective function of (\ref{prob: peak}),
    \begin{equation}
        J(c) := -pc+\sum_{t=0}^{T-1} h_t(v_t^*)-p_t^+[v_t^*-g_t]^+ + p_t^-[v_t^*-g_t]^-,
    \end{equation} where $v_t^*$ depends on $c$ through the optimization. Then $J(c)$ is concave in $c$.
\end{lemma}
\begin{proof}
    Let $I(c) = \{t: v_t^* = c+g_t\}$ be the set of time indices where the peak demand constraint is active. We can decompose the objective function as
    \begin{equation*}
    \begin{aligned}
        J(c)=&-pc + \sum_{t\in I(c)} h_t(c+g_t) - p_t^+ c \\&+\sum_{t\in [T-1]\setminus I(c)}h_t(v_t^*)-p_t^+[v_t^*-g_t]^++p_t^-[v_t^*-g_t]^-.
    \end{aligned}
    \end{equation*}
    As $c$ increases, the constraint set becomes less restrictive, so $|I(c)|$ is non-increasing. For the terms not in $I(c)$, the optimal $v_t^*$ is independent of $c$ and determined by the unconstrained optimum. 
    Taking the first and second derivatives,
    \begin{align*}
        &J'(c) = -p+\sum_{t\in I(c)} h'_t(c+g_t) - p_t^+,\\
        &J''(c) = \sum_{t\in I(c)}h''_t(c+g_t)\le 0,
    \end{align*}
    where the inequality follows from the concavity of $h_t$. Therefore, $J(c)$ is concave in $c$.
\end{proof}
The concavity of $J(c)$ ensures that the optimal $c^*$ is unique and can be found by solving the first-order condition $J'(c^*) = 0$. The key insight is that $I(c)$ changes only at discrete values of $c$, allowing us to search for the optimum efficiently. Specifically, define the candidate set
\begin{align*}
    \mathcal{C} = \{\min\{v_t^\dagger(g_t) - g_t, \bm{1}^\top\overline{\bm{d}}+\overline{e}-g_t\}\}_{t=0}^{T-1},
\end{align*}
where $v_t^\dagger(g_t)$ is the unconstrained optimum from (\ref{eq: uc_v}). The set $I(c)$ remains constant for $c\in[c_1,c_2]$, where
\begin{equation}\label{eq: interval}
\begin{aligned}
    &c_1=\max\{c: c\in \mathcal{C}, J'(c)< 0\},\\
    &c_2=\min\{c: c\in \mathcal{C}, J'(c)> 0\}.
\end{aligned}
\end{equation}
The optimal $c^*$ lies in this interval and satisfies
\begin{align}\label{eq: opt}
    J'(c^*)=-p+\sum_{t\in I(c_1)} h'_t(c^*+g_t)-p_t^+ = 0.
\end{align}
\subsection{Large Storage Peak Searching (LSPS) Algorithm}


The LSPS algorithm (Algorithm~\ref{alg: myopic}) is designed to solve the scheduling problem offline. For a given scheduling horizon ({\em e.g.,} 24 hours), the algorithm uses forecasts of exogenous variables, like renewable generation $g_t$, for the entire period to solve for a single optimal peak demand, $c^*$. This single $c^*$ is then used to determine the optimal control actions for each time step within the horizon. The algorithm efficiently approximates the problem (\ref{eq: prob}) by combining the relaxation techniques from the previous sections. It leverages the key structural properties established in the previous sections: the decomposition enabled by the large storage capacity relaxation, the analytical solution for fixed peak demand, and the concavity of the objective function in the peak demand parameter. By searching over the discrete candidate set $\mathcal{C}$, we avoid solving the full dynamic optimization problem and rather solve a linear complexity algorithm, as formalized below, while maintaining a small solution gap compared to the optimal. 
\begin{algorithm}[t]
\caption{Large storage peak searching (LSPS) algorithm}\label{alg: myopic}
\begin{algorithmic}[1]
\State Candidate set $\mathcal{C} \leftarrow \{\min\{v_t^\dagger(g_t) - g_t, \overline{e}+\bm{1}^\top\overline{\bm{d}}-g_t\}\}_{t=0}^{T-1}$
\State Derivatives $J'(\mathcal{C}) \leftarrow \{J'(c): c \in \mathcal{C}\}$
\State $c_1\leftarrow\max\{c: c\in \mathcal{C}, J'(c)< 0\}$
\State Solve $J'(c^*) = 0$ using (\ref{eq: opt}) for optimal peak demand $c^*$
\State Find optimal actions $v_t^*$ using (\ref{eq: myopic}) with peak demand $c^*$
\State Project $v_t^*$ back to variables $(e_t^*, \bm{d}_t^*)$ using (\ref{eq: helper})
\State Clip $e_t^*$ to ensure $s_{t+1}\in[0, B]$
\end{algorithmic}
\end{algorithm}

\begin{theorem}
Algorithm~\ref{alg: myopic} has computational complexity $\mathcal{O}(T)$.
\end{theorem}
\begin{proof}
Under the assumption that $h_t'$ and $(h_t')^{-1}$ can be evaluated in constant time, each step of the algorithm requires $\mathcal{O}(T)$ operations. Step 1 computes $T$ candidate values, Step 2 evaluates $T$ derivatives, Step 3 performs a linear search through $\mathcal{O}(T)$ values, Step 4 solves a single equation with less than $T$ terms in $\mathcal{O}(T)$ time, and Steps 5-7 compute $T$ optimal actions with each one constant time. Therefore, the overall complexity is $\mathcal{O}(T)$.
\end{proof}

Empirical analysis on the performance of the LSPS algorithm is presented in \sectsign\ref{sec:case_study}. A comparison between the computational cost of the LSPS algorithm and a deterministic optimization that solves (\ref{eq: prob}) offline in one-shot is presented in Appendix \ref{app:ComputationalCost}.

\section{Case Study}
\label{sec:case_study}
Using two rich datasets, and considering a one-day operation horizon with hourly decisions, $T = 24$ hours, we evaluate the DER co-optimization performance of the proposed LSPS algorithm and several benchmarks compared to a theoretical maximum that assumes perfect knowledge of future renewables solving (\ref{eq: prob}). Performance is measured by the solution gap, which is the difference between the cumulative reward of the algorithms and the theoretical upper bound. The experiments were run on a server with 1 AMD Ryzen Threadripper 3990X 64-Core Processor and 4 Nvidia RTX A4000 GPUs.


\subsection{Dataset Description}

\begin{figure}[th] 
    \centering
    \subfigure[Daily net demand data for HouseZero from Jun 2022 to May 2024]{
    \includegraphics[width=0.46\linewidth]{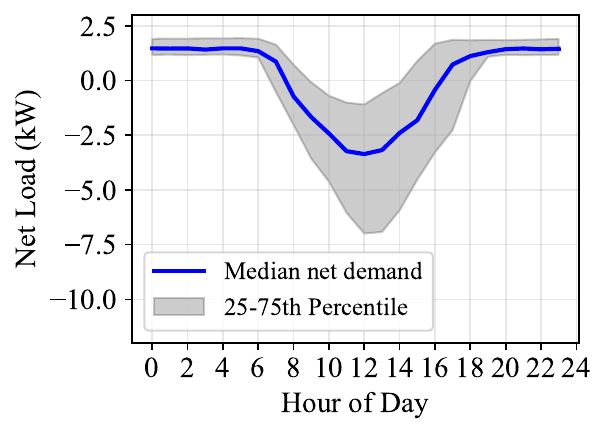}
    }
    \subfigure[Daily net demand data for BDGP building from Jan 2016 to Dec 2017]{
    \includegraphics[width=0.46\linewidth]{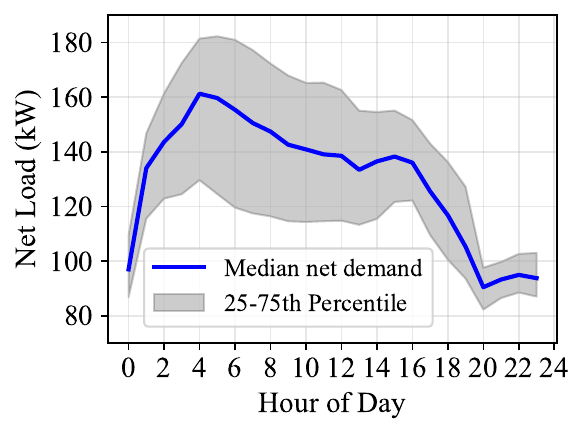}
    }
    \caption{Daily net demand distributions for HouseZero and BDGP over the training and testing period.}
    \label{fig:net_demand}
\end{figure}



Two open-source datasets from Harvard HouseZero~\cite{han_two-year_2024} and Building Data Genome Project (BDGP)~\cite{miller2020buildingdata} were used to evaluate the performance of DER scheduling under different algorithms. The two buildings were chosen with distinct characteristics, while both are equipped with solar PV. Renewable generation and demand data, resampled to hourly granularity, were used to train and test the performance of different control strategies. Fig. \ref{fig:net_demand} shows the net demand distribution for the HouseZero and BDGP buildings. The two buildings differ significantly in their energy pattern: HouseZero is a low-energy, ultra-efficient building with a peak demand below 2 kW and a frequently negative net consumption. In contrast, the BDGP building has substantially higher consumption, with net consumption peaking at over 160 kW. 



\subsection{Simulation Setup}
We adopt the concave utility function from \cite{samadi2012advanced}: $U_t(d_t)=\alpha_t d_t - \frac{1}{2}\beta_td_t^2$, where $\alpha_t$ and $\beta_t$are time-varying and learned a priori using historical energy prices, demands, and demand elasticity \cite{alahmed_integrating_2022}. The price elasticity of consumption was set to $-0.1$ \cite{asadinejad2018evaluation}. To get more crisp insights, we assumed time-invariant NEM tariff parameters, with a retail (import) rate of $p^+ = \$0.12$/kWh and a sell (export) rate of $p^- = \$0.06$/kWh. The daily peak demand charge was set to $p = \$10$/kW. The final salvage value for the remaining energy in the battery was set at $\gamma = \$0.09$/kWh. The baseline battery parameters were selected based on the demand data and the size of the building. For HouseZero, $B=5$ kWh and $\underline{e}=\overline{e}=1$ kW were used; while for BDGP building,  $B=350$ kWh and $\underline{e}=\overline{e}=50$ kW were used. In both cases, the battery charging/discharging efficiencies were set at $\tau=\rho=0.95$.

\subsection{Control Strategies}
We consider and compare five distinct control strategies for scheduling DER. The description of each strategy is provided below, while a detailed comparison of their respective advantages and limitations is presented in Appendix \ref{app:AlgoComparison}.
\paragraph{Backup Mode}
The backup mode prioritizes using renewables to charge the battery to maintain a high SoC level, which is usually set by the user. The battery discharges only during grid outages, equivalent to the case without a battery system \cite{jin_optimal_2017} during normal operations. Consequently, this control strategy cannot exploit the full potential of jointly optimizing demand and storage.


\paragraph{Renewable-Adjusted Threshold Policy (RATP)}
RATP is a threshold-based battery scheduling policy that adjusts charging and discharging based on the renewable adjusted gross consumption ($\tilde{d} = d - g$). When $\tilde{d}> 0$, the battery discharges to reduce grid purchases; when $\tilde{d}< 0$, it charges to avoid exporting excess renewables at the low $p^-$. When $\tilde{d}= 0$ the battery stays idle. Although simple, RATP is highly effective under NEM structures, as it arbitrages the difference between the retail and sell rates by minimizing power imports and exports. However,  RATP's performance significantly degrades under NEM with demand charges, as it is agnostic to strategic net consumption decisions that shape demand peaks.

Mathematically, the storage control under RATP, for a given demand $d_t= \hat{d}_t$, is
\begin{equation*}
    e_t(g_t)=\begin{cases}
        \max\{-\underline{e}, -\rho s_t, -\tilde{d}_t\}, & \tilde{d}_t(g)>0\\
        \min\{\overline{e}, (B-s_t)/\tau, -\tilde{d}_t\} & \tilde{d}_t(g)\le 0,
    \end{cases}
\end{equation*}
for $t=0,1,\ldots,T-1$.

\paragraph{RL Algorithm}
We use a popular model-free RL method of proximal policy optimization (PPO) algorithms to directly learn the optimal policy for controlling the DER system described. PPO uses an actor-network and a critic network with a multilayer perceptron (MLP) structure. For training, each of the actor and critic networks consists of two hidden layers with 256 neurons per layer. The actor-network predicts the optimal action, while the critic network estimates the optimal value function. We train the networks with a learning rate of $1e^{-4}$, batch size of 64, and 16 epochs per update. Other training parameters are kept as the default provided by Stable-baselines3 \cite{stable-baselines3}. Training histories for RL are provided in Appendix \ref{appendix:train_history}, showing the convergence of the trained policy that is used in \sectsign\ref{subsec:sim}.


\paragraph{Large Storage Peak Searching (LSPS) Algorithm} 
LSPS algorithm, described in Algorithm~\ref{alg: myopic}, offers an efficient and scalable method to co-optimize DER. Following Algorithm~\ref{alg: myopic}, we first calculate the optimal peak, then, by using Lemma \ref{lem:OptimalRelaxed} with the optimal peak, we compute the co-optimized battery and demand actions at each time step.

\paragraph{Theoretical Upper Bound}
The theoretical upper bound assumes perfect knowledge of the renewables for the entire scheduling period and deterministically solves the constrained convex program in (\ref{eq: prob}) in one-shot while relaxing the simultaneous charging-discharging constraint, giving it a slight advantage over LSPS.

During the prediction stage, RL takes in the current state values, including battery SoC,  renewable generation, peak demand, and outputs the optimal actions based on the learned policy. While the LSPS algorithm during prediction takes in the optimal peak demand in addition to the current demand and generation. The optimal actions are generated by solving an optimization problem. During the training stage, a similar level of computation is required as RL learns the system dynamics, including the patterns of the renewable profile with two years of hour-level training data, while LSPS requires models for renewable predictions.


\subsection{Simulation Results}\label{subsec:sim}

We compare the performance of the control strategies (a)-(d) above to the theoretical upper bound in (e) by varying the following parameters:

\begin{itemize}
    \item \textbf{Battery capacity ($B$)}: 5 kWh -- 50 kWh for HouseZero and 350 kWh -- 950 kWh for BDGP building to simulate battery sizes from small to large.
    \item \textbf{Salvage value rate ($\gamma$)}: \$0.03$/\text{kWh}$ -- \$15$/\text{kWh}$ for cases including $\gamma < p_t^{-}$, $p_t^{-} < \gamma < p_t^{+}$, and $\gamma > p_t^{+}$.
    \item \textbf{Export rate of electricity ($p^{-}$)}: $\$0/\text{kWh}$ 
    -- $\$0.12/\text{kWh}$ from no to high sell back incentive with $p^{-}=p^{+}$.
    \item \textbf{Peak demand price of electricity ($p$)}: $\$0/\text{kW}$ -- $\$10/\text{kW}$ from no penalty on a high peak to a reasonably high peak charge commonly practiced in electricity tariffs.
\end{itemize}

Table~\ref{table:result_summary} summarizes the overall performance comparisons in terms of the solution gap to the optimal solutions.\footnote{Table~\ref{table:result_all} in Appendix \ref{appendix:all_result} provides detailed results for all cases.} Overall, the LSPS algorithm performed closest to the optimal solutions over all scenarios with an average gap of 12.86\%, followed by RL of 20.21\%, RATP algorithm of 41.93\%, and backup mode of 43.75\%. The LSPS algorithm was closest to the optimal 34 out of 42 scenarios (81\%) while RL was closest to optimal in the remaining 8 scenarios (19\%). Furthermore, for almost all cases, control algorithms showed better performance for the smaller building of HouseZero, with LSPS performing only 4.52\% below the theoretical optimal. 

The performance gap between the HouseZero and BDGP buildings is mainly attributed to differences in their peak demand, total energy consumption, and net load profiles. Specifically, the shape of the net load depends on whether the building transitions between net-producing and net-consuming regimes, the timing of these transitions, and the duration of each regime. In the case of HouseZero, lower demand and higher on-site generation result in negative net load during daytime hours, with peak demand occurring in the early morning or late evening. In contrast, the substantially higher demand of the BDGP building prevents it from reaching a net-producing state, leading to peak demand periods during daytime hours.
        
\begin{table}[t]
\centering
\setlength{\tabcolsep}{6pt}
\caption{Summary of solution gap (\%) of different control strategies}
\label{table:result_summary}
\resizebox{\columnwidth}{!}{
\begin{tabular}{cccccc}
\toprule
\midrule
Case & Building & \multicolumn{4}{c}{Percentage gaps to optimal} \\
\cmidrule(r){1-2}  \cmidrule(r){3-6}
&  & Backup & RATP & RL & LSPS \\
\midrule
\multirow{2}{*}{Battery capacity} 
                & HouseZero & 47.44\% & 14.17\% & 9.33\% & 1.03\%\\
                & BDGP & 64.42\% & 63.56\% & 29.50\% & 25.37\% \\
\midrule
\multirow{2}{*}{Salvage value rate} 
                & HouseZero & 41.06\% & 52.41\% & 19.54\% & 6.06\% \\
                & BDGP & 50.16\% & 66.56\% & 29.28\% & 18.64\%\\
\midrule
\multirow{2}{*}{Export rate} 
                & HouseZero & 50.83\% & 46.21\% & 22.58\% & 8.59\% \\
                & BDGP & 62.23\% & 61.26\% & 32.76\% & 31.91\% \\
\midrule
\multirow{2}{*}{Peak demand price} 
                & HouseZero & 17.86\% & 15.85\% & 8.01\% & 2.41\% \\
                & BDGP & 16.04\% & 15.39\% & 10.70\% & 8.86\% \\

\midrule
\multirow{2}{*}{\textbf{Overall}} 
                & HouseZero & 39.30\% & 32.16\% & 14.86\% & 4.52\% \\
                & BDGP & 48.21\% & 51.69\% & 25.56\% & 21.20\% \\


\midrule
\bottomrule
\end{tabular}
}
\end{table}

\begin{figure*}[th]
    \centering
    
    \subfigure[HouseZero, battery capacity case]{
        \label{subfig:hz_battery_case}
        \includegraphics[width=0.23\linewidth]{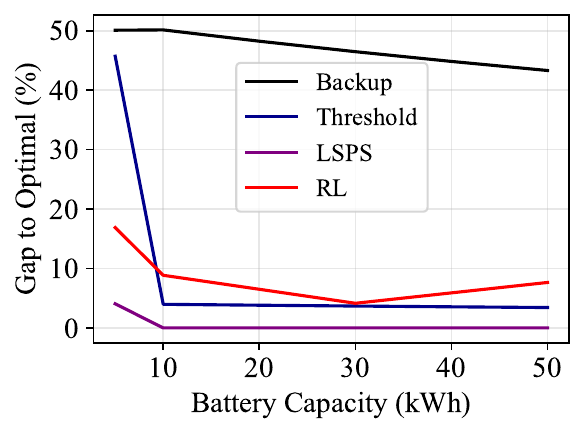}
    }
    \subfigure[HouseZero, salvage value rate case]{
        \label{subfig:hz_salvage_case}
        \includegraphics[width=0.23\linewidth]{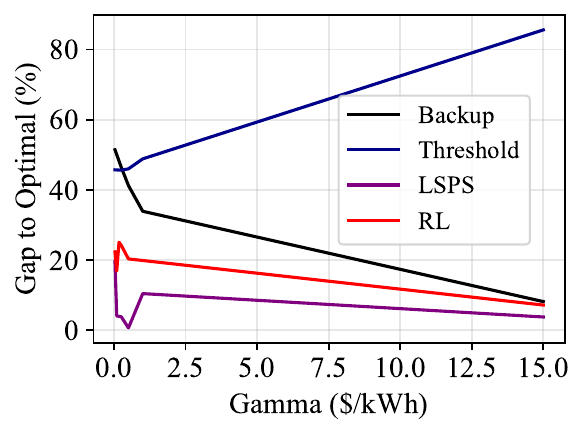}
    }
    \subfigure[HouseZero, export rate case]{
        \label{subfig:hz_export_case}
        \includegraphics[width=0.23\linewidth]{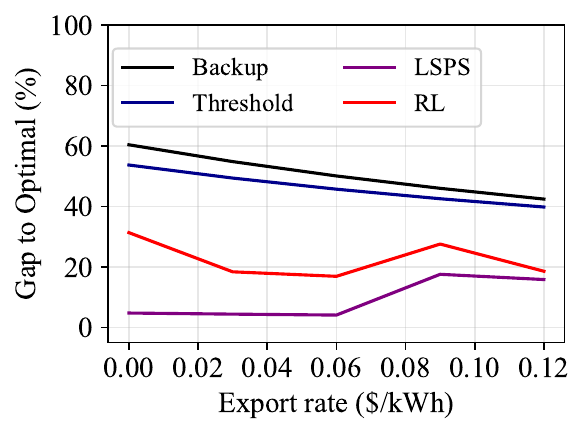}
    }
    \subfigure[HouseZero, peak charge case]{
        \label{subfig:hz_peak_case}
        \includegraphics[width=0.23\linewidth]{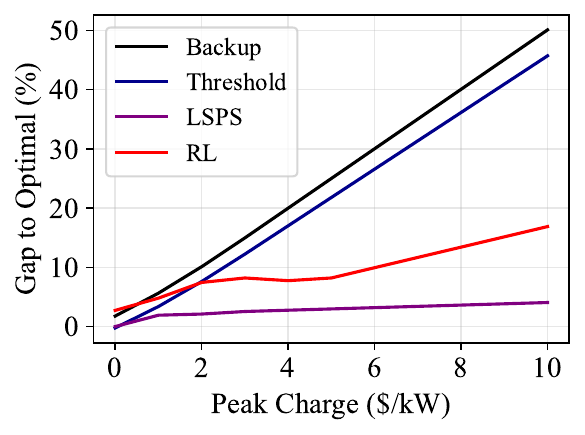}
    }
    
    \subfigure[BDGP, battery capacity case]{
        \label{subfig:bdgp_battery_case}
        \includegraphics[width=0.23\linewidth]{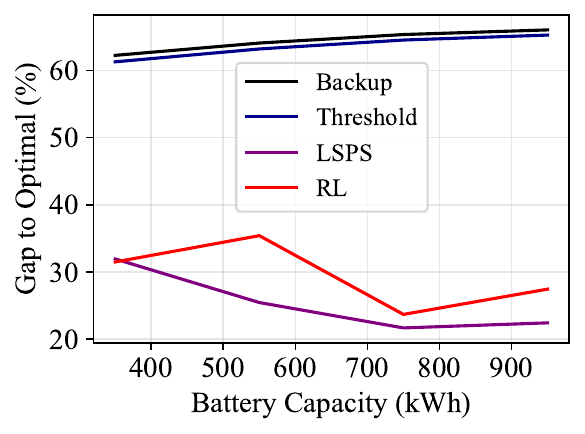}
    }
    \subfigure[BDGP, salvage value rate case]{
        \label{subfig:bdgp_salvage_case}
        \includegraphics[width=0.23\linewidth]{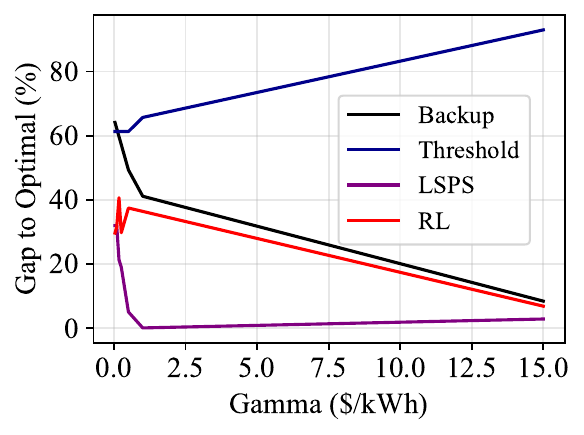}
    }
    \subfigure[BDGP, export rate case]{
        \label{subfig:bdgp_export_case}
        \includegraphics[width=0.23\linewidth]{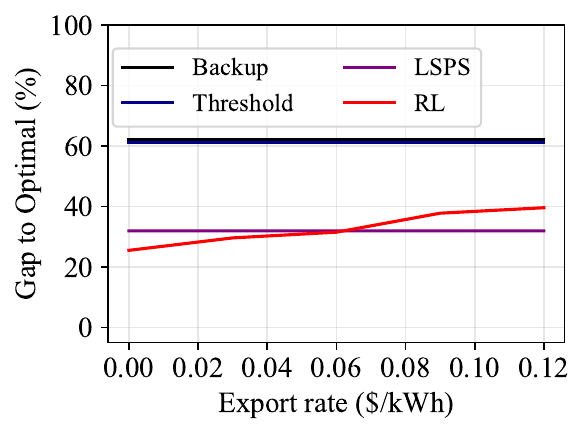}
    }
    \subfigure[BDGP, peak charge case]{
        \label{subfig:bdgp_peak_case}
        \includegraphics[width=0.23\linewidth]{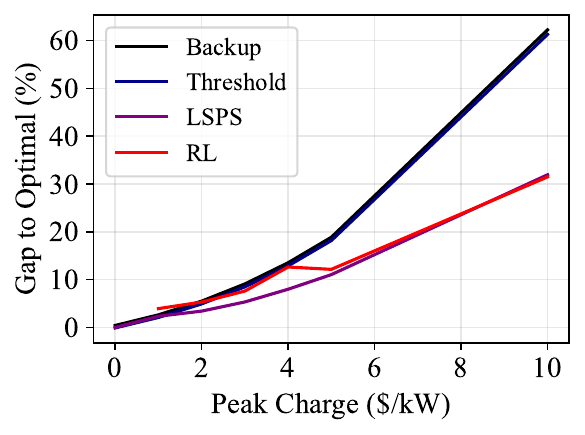}
    }
    
    \caption{Comparison of daily surplus gaps to theoretical upper bound for different control strategies with varying parameters using HouseZero (top row) and BDGP (bottom row) test data. The cases include (a,e) battery capacities, (b,f) salvage value rates, (c,g) electricity export rates, and (d,h) peak demand prices.}
    \label{fig:combined_all_cases}
\end{figure*}


Illustrated in Fig.~\ref{fig:combined_all_cases}, for the first group of test cases varying battery capacity, the LSPS algorithm performed closest to the theoretical upper bound for both HouseZero and BDGP in Fig.~\ref {subfig:hz_battery_case} and Fig.~\ref {subfig:bdgp_battery_case}. For the HouseZero case with low loads and large battery capacities, the peak was completely removed (0 kW), leading the LSPS algorithm to perform as well as the theoretical optimal. Large battery capacity also significantly improved the RATP algorithm performance, even though it significantly fell behind the RL and LSPS algorithms when the battery size was more limited. On the other hand, RL followed a similar trend as the LSPS algorithm while performing around 10\% worse over all cases. RL was still stable and more adaptive to changes in the environment than the RATP and backup algorithms. For the much larger BDGP building, both RL and LSPS algorithms improved overall with the increased battery capacity, while backup and RATP failed to show the expected improvement due to their inflexibility in considering the high peak and energy use in this building. 

Second, varying salvage value rate again showed the superior performance of the LSPS algorithm for both HouseZero and BDGP in Fig. \ref{subfig:hz_salvage_case} and Fig. \ref{subfig:bdgp_salvage_case}. The RATP algorithm tended to use up the available battery charge, leading to low salvage value at the end of the day and causing a larger gap to optimal with an increased salvage value rate. Meanwhile, for extremely high salvage value cases, the backup case got closer to optimal as the strategy kept the battery full at the end of the day. RL performed relatively well and was closer to the LSPS algorithm and optimal for the relatively smaller HouseZero building. Overall, for both HouseZero and BDGP, the performances exhibited a similar trend across the four control algorithms.
 
Third, varying export rates highlighted the difference between HouseZero and BDGP buildings, illustrated in Fig. \ref{subfig:hz_export_case} and Fig. \ref{subfig:bdgp_export_case}. For the BDGP building, the higher energy demand meant the building would never be in a net producing regime, leading to flat lines for the control algorithms. With RL, there were slight variations in performance due to the stochasticity during training. For HouseZero, all control algorithms improved with a higher export rate since it generated more revenue for the building to sell back the excess energy production from solar. LSPS and RL algorithms performed relatively well and were stable across different export rates over the two test buildings.

Lastly, varying peak demand charges saw all control algorithms change in a similar trend shown in Fig. \ref{subfig:hz_peak_case} and Fig. \ref{subfig:bdgp_peak_case}. It was observed that for low peak charges, all control algorithms performed well, highlighting the complexity of peak demand charges added to optimally managing such a DER system. With higher peak demand charges, all control algorithms diverged from optimal, with LSPS and RL algorithms performing the best across different scenarios. Additionally, both of them showed a smaller solution gap for the smaller building of HouseZero, indicating the challenges for buildings with higher demand and energy consumption.

In addition to analyzing the sensitivity of the above-mentioned parameters, we compared the performance of the RL and LSPS algorithms under different load and renewable patterns in Tab. \ref{tab:diff_load_cases}. Results show that higher generation and lower demand provided higher surplus and a smaller gap to optimal from both RL and LSPS. On average, RL performed with a gap of 17.93\% while LSPS performed with a gap of 4.43\%. The results demonstrated the effectiveness of the RL and LSPS across various load patterns, highlighting the potential for implementing the algorithm in future long-term real-world tests.


\begin{table}[h]
    \centering
    \caption{RL and LSPS performance under different generation and demand scenarios.}
    \label{tab:diff_load_cases}
    \resizebox{\columnwidth}{!}{
    \begin{tabular}{ccccccc}
        \toprule
        \midrule
        & & \multicolumn{3}{c}{Surplus} & \multicolumn{2}{c}{Gap to Optimal} \\
        \cmidrule(r){1-2}  \cmidrule(r){3-5} \cmidrule(r){6-7}
        Generation & Demand & RL & LSPS & Optimal & RL & LSPS \\
        \midrule
        75th Percentile & 25th Percentile & 16.4 & 18.9 & 20.3 & 19.2\% & 6.9\% \\
        50th Percentile & 50th Percentile & 18.2 & 21.0 & 21.9 & 16.9\% & 4.1\% \\
        25th Percentile & 75th Percentile  & 18.1 & 21.5 & 22.0 & 17.7\% & 2.3\% \\
        \midrule
        \bottomrule
    \end{tabular}
    }
\end{table}

Overall, both LSPS and RL showed great robustness and performance across different scenarios with varying battery capacities, salvage values, and prices of electricity. Specifically, LSPS optimization outperformed backup and RATP algorithms by 68.47\% and 68.9\%, while RL outperformed them by 52.49\% and 48.01\%, respectively. LSPS and RL algorithms trended in sync in their gap to optimality, demonstrating their ability to capture system dynamics, while the LSPS algorithm is superior in around 81\% instances. Moreover, LSPS was optimal with a large battery capacity and close to optimal for a low export rate and peak demand charge. With 0 or low peak demand charges, all control algorithms performed close to optimal. In comparison, the RATP and backup algorithms were sensitive to changes in battery capacity, salvage rate, peak demand charge, and export rate when the building entered the net-producing regime. For cases where the battery capacity was large or the peak price was low, RATP outperformed RL and approached optimal with deterministic predictions. However, with a high salvage value and demand price, the RATP algorithm's performance degraded significantly and became even worse than the backup mode. Backup mode varied mostly linearly in its gap to the optimal case. It performed well only in special cases with high salvage value rates and low peak demand cases.


\begin{figure}[th] 
    \centering
    \subfigure[Battery action comparison]{
    \includegraphics[width=0.46\linewidth]{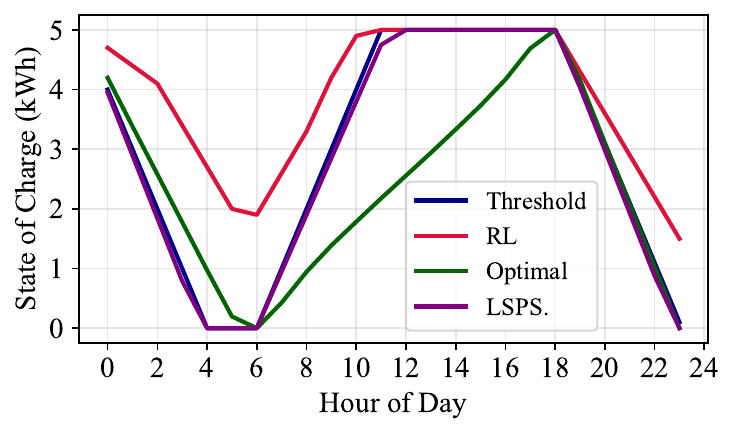}
    }
    \subfigure[Demand action comparison]{
    \includegraphics[width=0.46\linewidth]{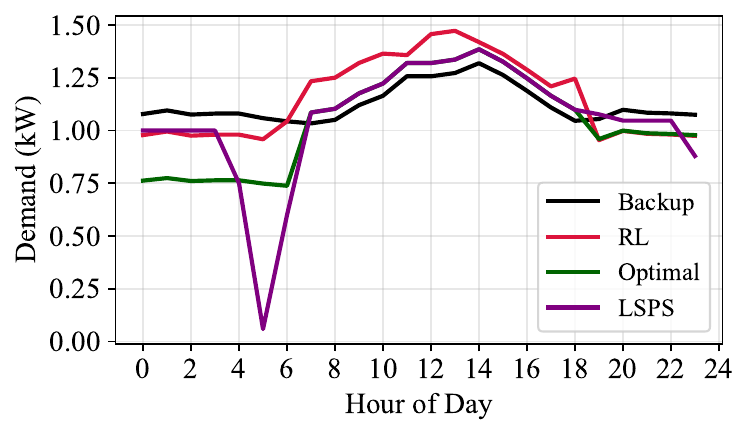}
    }
    \caption{Typical day battery and demand actions comparison between different control methods for test building HouseZero.}
    \label{fig:typical_behavior}
\end{figure}

As an example, Fig.~\ref{fig:typical_behavior} shows the action sequence for the baseline configuration of HouseZero over one test day. All strategies followed a similar behavior, which discharged the battery at the beginning and end of the day. At the same time, demand was reduced in the evening when there was no solar generation. Both RATP and RL were more aggressive in battery charging actions when lacking generation predictions. The optimal battery actions indicated a much slower charging behavior over the daytime with precise knowledge of solar generation, creating the performance gap mentioned above. It was also observed that the LSPS algorithm may generate aggressive demand reduction actions due to the lack of predictions. Such behavior negatively impacted the utility function and the overall objective, causing larger solution gaps.

\section{Conclusion}
\label{sec:conclusion}

This work provides the first analytical characterization of the joint optimization of flexible demand and BESS under demand charges and bi-directional power flow. We model the problem using a stochastic dynamic program and show that the optimal policy exhibits a threshold structure. To address computational complexity, we develop an efficient demand peak searching approximation method, based on a relaxed version of the problem, leading to a linear-complexity algorithm that identifies the optimal peak demand.

Through simulations with real data from two rich and distinct datasets, we demonstrated the performance of the LSPS algorithm in achieving a small solution gap compared to baseline approaches, highlighting its effectiveness compared to relatively more computationally prohibitive methods such as model-predictive control and RL.

The work presented here has several limitations that open avenues for future research. Firstly, the RL algorithm used in our simulations is relatively simple. Designing a more sophisticated algorithm that reliably converges to the theoretical upper bound remains an important direction for future work. Second, the LSPS algorithm assumes full knowledge of the utility’s functional form $U(\cdot)$, which may be difficult to obtain in practice. This limitation can be addressed using machine learning or RL techniques that infer the utility function from observed behavior \cite{Ng&Russell:20ICML}. Lastly, it is worthwhile to investigate whether sufficient conditions can be established under which the LSPS algorithm remains optimal, even when projections onto the feasible solution set are involved.


\bibliographystyle{IEEEtranDOI}
\bibliography{ref}


\appendix
\subsection{Additional Theoretical Results}\label{app:theory}

\begin{theorem}\label{thm: myopic}
    The optimal policy of (\ref{eq: prob}) is not myopic when the constraint (\ref{DCevolution}) or (\ref{SoCevolution}) is relaxed.
\end{theorem}
\begin{proof}
    When (\ref{SoCevolution}) is relaxed, suppose the optimal policy is myopic. Consider a special case where horizon $T=2$, initial SoC $s>2\underline{e}/\rho$, and $U_0=U_1=U$, $p^+_0=p_1^+=p^+$, $p^-_0=p_1^-=p^-$ are static. Then the optimization is
    \begin{subequations}
        \begin{align}
            \max_{v_0, v_1} \quad &h(v_0)+h(v_1) - p^+([v_0-g_0]^+ +[v_1-g_1]^+)\nonumber\\
            &+p^-(([v_0-g_0]^- +[v_1-g_1]^-)\nonumber\\
            &-p\max\{v_0-g_0, v_1-g_1, 0\}\\
            \text{s.t.}\quad & v_0, v_1\in[-\underline{e}, \bar{e}+\bar{d}]
        \end{align}
    \end{subequations}
   Let $g>0$ satisfies $h'(g)=p^++p$, then on case $g_0=g_1=g$, $v_0^*=v_1^*=g$. By the assumption, the optimal policy is myopic, so $v_0^*=g$ once $g_0=g$. However, when $g_0=g>g_1$, the optimal solution is $(g+g_0-g_1, g)$, leading to a contradiction. So the optimal policy cannot be myopic when only (\ref{SoCevolution}) is relaxed. 

    When (\ref{DCevolution}) is relaxed, suppose the optimal policy is myopic. Consider a special case where the horizon $T=2$, initial SoC $s=\underline{e}$, $U_0=U_1$ are constant functions, $\tau=\rho=1$, $\gamma=0$. Then the optimization is 
    \begin{subequations}
        \begin{align}
            \max_{e_0, e_1} \quad & - p_0^+[e_0-g_0]^+ + p_0^-[e_0-g_0]^-\nonumber\\
            &-p_1^+[v_1-g_1]^+ + p_1^-[v_1-g_1]^-\\
            \text{s.t.}\quad & s_1=\underline{e}+e_0\in[0, B],\\ &s_2=\underline{e}+e_0+e_1\in[0, B].
        \end{align}
    \end{subequations}
    When $g_0=g_1=0$, $p_0^+,p_1^+>p_0^->p_1^-$, the optimal solution is $(e_0, e_1)=(-\underline{e}, 0)$. By assumption, $e_0^*=-\underline{e}$ is independent of $p_1^-$ since the optimal policy is myopic. However, when $g_0=g_1=0$, $p_0^+,p_1^+>p_1^->p_0^-$, the optimal solution is $(e_0, e_1)=(0, -\underline{e})$, contradicting to $e_0^*=-\underline{e}$. So the optimal policy cannot be myopic when only (\ref{DCevolution}) is relaxed.
\end{proof}

\subsection{Additional Simulation Results}
\label{appendix:all_result}

Table \ref{table:result_all} records all the test cases performed comparing Backup mode, RATP (threshold), LSPS Algorithm, RL, and theoretical upper bound results. For each of the test scenarios, we varied one of the four parameters for both the HouseZero and BDGP buildings. For each case, the control strategy that performed the best (closest to optimal) is shown in bold. The total surplus (utility-electricity cost) for the five algorithms is listed. 

In general, LSPS achieved the lowest solution gap, being closest to optimal in 34 of 42 scenarios (81\%), while RL is close to optimal in the remaining 8 scenarios (19\%). For the smaller HouseZero building, LSPS outperforms RL in all cases. The numerical results highlight the inability of the backup and RATP to adapt to changing environmental settings. We observed significant performance degradation for cases with large battery capacity, high salvage value, and high demand charge. Meanwhile, LSPS and RL better adapt to the different scenarios while maintaining reasonable performance.

\begin{table*}[!htbp] 
\centering
\setlength{\tabcolsep}{6pt}
\caption{Surplus for different methods under all test scenarios with varying battery capacities, salvage values, sell prices, and peak prices.}
\label{table:result_all}
\begin{tabular}{cccccccccc}
\toprule
\midrule
\multicolumn{5}{c}{Scenario} & \multicolumn{4}{c}{Surplus for different methods} \\
\cmidrule(r){1-5}  \cmidrule(r){6-10}
Battery Capacity & Salvage Value & Sell Price & Peak Price & (Dis)Charge Limits & Backup & Threshold & LSPS & RL & Optimal \\ 
\midrule
\midrule
5kWh & \$ 0.09/kWh & \$ 0.06/kWh & \$ 10/kW & 1kW & 10.9 & 11.9 & \textbf{21.0} & 18.2 & 21.9\\
10kWh & \$ 0.09/kWh & \$ 0.06/kWh & \$ 10/kW & 1kW & 11.4 & 21.9 & \textbf{22.8} & 20.2 & 22.8\\
30kWh & \$ 0.09/kWh & \$ 0.06/kWh & \$ 10/kW & 1kW & 13.2 & 23.7 & \textbf{24.6} & 23.6 & 24.6\\
50kWh & \$ 0.09/kWh & \$ 0.06/kWh & \$ 10/kW & 1kW & 15.0 & 25.5 & \textbf{26.4} & 24.4 & 26.4\\
\midrule
350kWh & \$ 0.09/kWh & \$ 0.06/kWh & \$ 10/kW & 50kW & 419 & 429.7 & 755.2 & \textbf{760.2} & 1109.6 \\
550kWh & \$ 0.09/kWh & \$ 0.06/kWh & \$ 10/kW & 50kW & 437 & 447.7 & \textbf{906.8} & 785.7 & 1216.3 \\
750kWh & \$ 0.09/kWh & \$ 0.06/kWh & \$ 10/kW & 50kW & 455 & 465.7 & 973.9 & \textbf{1028.1} & 1312.5 \\
950kWh & \$ 0.09/kWh & \$ 0.06/kWh & \$ 10/kW & 50kW & 473 & 483.7 & \textbf{1080.1} & 1010.2 & 1392.2 \\
\midrule
\midrule
5kWh & \$ 0.03/kWh & \$ 0.06/kWh & \$ 10/kW & 1kW & 10.6 & 11.9 & \textbf{17.6} & 17.0 & 21.9 \\
5kWh & \$ 0.09/kWh & \$ 0.06/kWh & \$ 10/kW & 1kW & 10.9 & 11.9 & \textbf{21.0} & 18.2 & 21.9 \\
5kWh & \$ 0.17/kWh & \$ 0.06/kWh & \$ 10/kW & 1kW & 11.3 & 11.9 & \textbf{21.0} & 16.4 & 21.9 \\
5kWh & \$ 0.25/kWh & \$ 0.06/kWh & \$ 10/kW & 1kW & 11.7 & 11.9 & \textbf{21.1} & 16.6 & 21.9 \\
5kWh & \$ 0.5/kWh & \$ 0.06/kWh & \$ 10/kW & 1kW & 13.0 & 11.9 & \textbf{21.9} & 17.4 & 22.1 \\
5kWh & \$ 15/kWh & \$ 0.06/kWh & \$ 10/kW & 1kW & 85.5 & 13.4 & \textbf{89.6} & 86.4 & 93.3 \\
\midrule
350kWh & \$ 0.03/kWh & \$ 0.06/kWh & \$ 10/kW & 50kW & 398.0 & 429.7 & 755.2 & \textbf{781.1} & 1109.6 \\
350kWh & \$ 0.09/kWh & \$ 0.06/kWh & \$ 10/kW & 50kW & 419.0 & 429.7 & 755.2 & \textbf{760.2} & 1109.6 \\
350kWh & \$ 0.17/kWh & \$ 0.06/kWh & \$ 10/kW & 50kW & 447.0 & 429.7 & \textbf{873.9} & 658.6 & 1109.6 \\
350kWh & \$ 0.25/kWh & \$ 0.06/kWh & \$ 10/kW & 50kW & 475.0 & 429.7 & \textbf{899.3} & 780.0 & 1109.6 \\
350kWh & \$ 0.5/kWh & \$ 0.06/kWh & \$ 10/kW & 50kW & 562.5 & 429.7 & \textbf{1054.3} & 694.3 & 1109.6 \\
350kWh & \$ 15/kWh & \$ 0.06/kWh & \$ 10/kW & 50kW & 5637.5 & 429.7 & \textbf{5980.8} & 5734.2 & 6152.0 \\
\midrule
\midrule
5kWh & \$ 0.09/kWh & \$ 0/kWh & \$ 10/kW & 1kW & 7.4 & 8.7 & \textbf{17.9} & 12.9 & 18.8 \\
5kWh & \$ 0.09/kWh & \$ 0.03/kWh & \$ 10/kW & 1kW & 9.2 & 10.3 & \textbf{19.4} & 16.6 & 20.3 \\
5kWh & \$ 0.09/kWh & \$ 0.06/kWh & \$ 10/kW & 1kW & 10.9 & 11.9 & \textbf{21.0} & 18.2 & 21.9 \\
5kWh & \$ 0.09/kWh & \$ 0.09/kWh & \$ 10/kW & 1kW & 12.7 & 13.5 & \textbf{19.3} & 17.0 & 23.5 \\
5kWh & \$ 0.09/kWh & \$ 0.12/kWh & \$ 10/kW & 1kW & 14.4 & 15.1 & \textbf{21.1} & 20.4 & 25.1 \\
\midrule
350kWh & \$ 0.09/kWh & \$ 0/kWh & \$ 10/kW & 50kW & 419.0 & 429.7 & 755.2 & \textbf{826.9} & 1109.2 \\
350kWh & \$ 0.09/kWh & \$ 0.03/kWh & \$ 10/kW & 50kW & 419.0 & 429.7 & 755.2 & \textbf{781.2} & 1109.2 \\
350kWh & \$ 0.09/kWh & \$ 0.06/kWh & \$ 10/kW & 50kW & 419.0 & 429.7 & 755.2 & \textbf{760.2} & 1109.2 \\
350kWh & \$ 0.09/kWh & \$ 0.09/kWh & \$ 10/kW & 50kW & 419.0 & 429.7 & \textbf{755.2} & 690.5 & 1109.2 \\
350kWh & \$ 0.09/kWh & \$ 0.12/kWh & \$ 10/kW & 50kW & 419.0 & 429.7 & \textbf{755.2} & 670.4 & 1109.2 \\
\midrule
\midrule
5kWh & \$ 0.09/kWh & \$ 0.06/kWh & \$ 1/kW & 1kW & 20.8 & 21.3 & \textbf{21.6} & 21.0 & 22.1 \\
5kWh & \$ 0.09/kWh & \$ 0.06/kWh & \$ 2/kW & 1kW & 19.7 & 20.3 & \textbf{21.5} & 20.3 & 21.9 \\
5kWh & \$ 0.09/kWh & \$ 0.06/kWh & \$ 3/kW & 1kW & 18.6 & 19.2 & \textbf{21.3} & 20.1 & 21.9 \\
5kWh & \$ 0.09/kWh & \$ 0.06/kWh & \$ 4/kW & 1kW & 17.5 & 18.2 & \textbf{21.3} & 20.2 & 21.9 \\
5kWh & \$ 0.09/kWh & \$ 0.06/kWh & \$ 5/kW & 1kW & 17.1 & 16.4 & \textbf{21.2} & 20.1 & 21.9 \\
5kWh & \$ 0.09/kWh & \$ 0.06/kWh & \$ 10/kW & 1kW & 10.9 & 11.9 & \textbf{21.0} & 18.2 & 21.9 \\
\midrule
350kWh & \$ 0.09/kWh & \$ 0.06/kWh & \$ 1/kW & 50kW & 1868.5 & 1879.0 & \textbf{1874.2} & 1843.5 & 1919.0 \\
350kWh & \$ 0.09/kWh & \$ 0.06/kWh & \$ 2/kW & 50kW & 1707.4 & 1717.9 & \textbf{1745.2} & 1711.0 & 1806.9 \\
350kWh & \$ 0.09/kWh & \$ 0.06/kWh & \$ 3/kW & 50kW & 1546.4 & 1556.9 & \textbf{1610.5} & 1571.8 & 1701.2 \\
350kWh & \$ 0.09/kWh & \$ 0.06/kWh & \$ 4/kW & 50kW & 1385.3 & 1395.8 & \textbf{1474.1} & 1399.6 & 1602.1 \\
350kWh & \$ 0.09/kWh & \$ 0.06/kWh & \$ 5/kW & 50kW & 1224.3 & 1234.8 & \textbf{1342.3} & 1325.7 & 1509.0 \\
350kWh & \$ 0.09/kWh & \$ 0.06/kWh & \$ 10/kW & 50kW & 419.0 & 429.7 & 755.2 & \textbf{760.2} & 1109.2 \\
\midrule
\bottomrule
\end{tabular}
\end{table*}

\subsection{Computational Cost Comparison Between LSPS and Deterministic Optimization}\label{app:ComputationalCost}

LSPS optimizes one step while deterministic optimization's complexity depends on the frequency of action (decision period) and length of planning horizon. Fig. \ref{fig:comp_cost_compare} compares at different length of planning horizon, the computational time for both algorithms to generate one-step control actions and also to compute all actions through the planning horizon.

The main advantage of LSPS lies in the fact that it only needs to solve for optimal actions at each time step during control, whereas deterministic optimization must solve for the entire planning horizon. Thus, as the planning horizon increases, the computational time grows and becomes much less efficient than LSPS. While it is typically more efficient if all actions throughout the planning horizon are required to be solved, in practice and in real-world implementations, such a scenario is not likely needed.

\begin{figure*}[th]
    \centering
    \subfigure[Computational time for each step]{
    \includegraphics[width=0.44\linewidth]{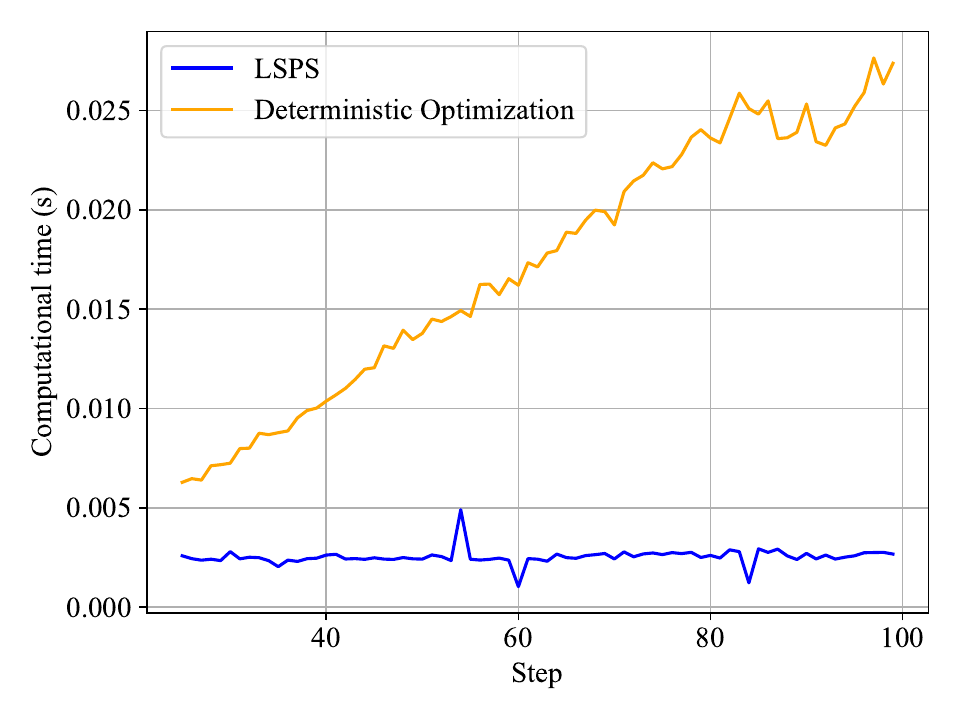}
    }
    \subfigure[Computational time for entire horizon]{
    \includegraphics[width=0.44\linewidth]{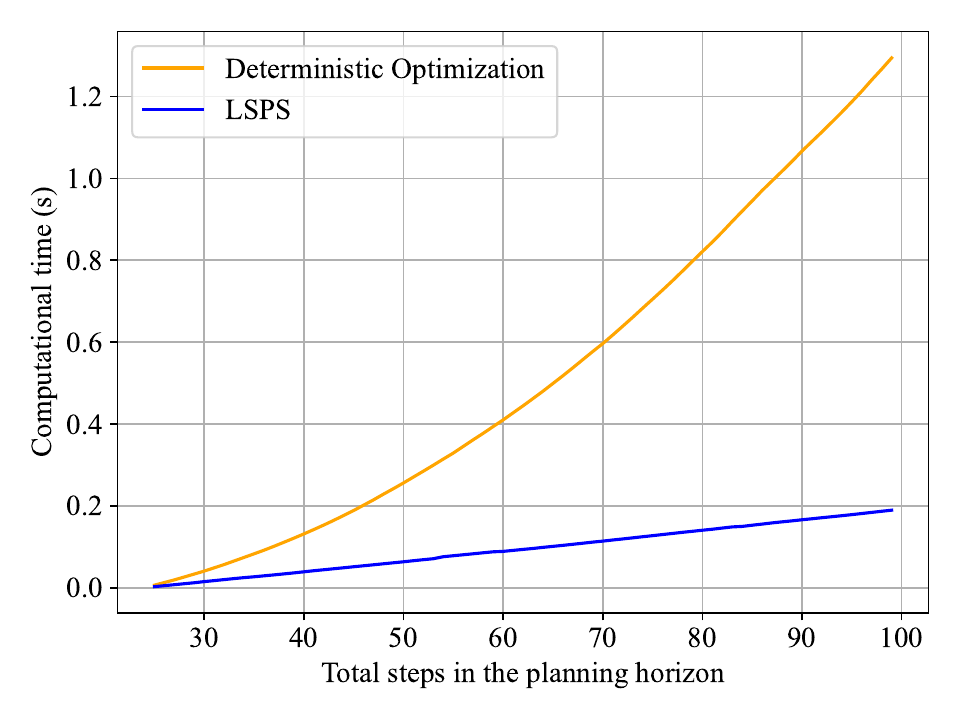}
    }
    \caption{Computational time comparison between LSPS and deterministic optimization.}
    \label{fig:comp_cost_compare}
\end{figure*}

\subsection{Comparison of Control Strategies}\label{app:AlgoComparison}

The advantages and weaknesses of the five different control strategies are summarized in Table \ref{tab:algo_ad_disad}. The \textit{optimal} control strategy achieves a theoretical upper bound that is obtained by solving a deterministic optimization problem. \textit{Backup} strategy is simple yet provides no additional operational savings, as it prioritizes storing power for contingencies. This control strategy is similar to Tesla's Powerwall backup operation mode. The renewable-adjusted threshold policy (RATP) algorithm is robust and safe with predefined rules that are highly suitable for NEM tariffs, yet it lacks adaptiveness to changes in parameters and environment, in addition to being agnostic to demand charges. RL and LSPS algorithms are more advanced and can adapt to different parameters and environments, which was reflected in their improved solution gap compared to the other algorithms. While RL requires training and LSPS requires renewable forecast, they both maintain a reasonable level of complexity during the prediction stage.

\begin{table*}[h]
    \centering
    \caption{Advantages and weaknesses of the considered control strategies}
    \label{tab:algo_ad_disad}
    \resizebox{\textwidth}{!}{
    \begin{tabular}{lll}
        \toprule
        \midrule
        Control strategies & Advantages & Weaknesses \\
        \cmidrule(r){1-1} \cmidrule(r){2-3}
        Optimal & Maximum surplus & \makecell[l]{Requires perfect knowledge of renewable generation \\ Exponential complexity with respect to the state of the dynamic program} \\
        \midrule
        Backup & \makecell[l]{No/low computation cost \\ Maximum storage during blackouts} & \makecell[l]{No cost savings and peak demand reduction \\ Energy loss over long-term storage}  \\
        \midrule
        RATP & \makecell[l]{Low computation cost \\ Robust and safe \\ Close to optimal for special cases} & \makecell[l]{Demand charge is not considered in the policy \\ Nonadaptive to changes in parameters \\ Negatively affect performance under special cases } \\
        \midrule
        RL & \makecell[l]{Close to optimal \\ Adaptive to different environment settings \\ Model-free, does not need prediction models} & \makecell[l]{Data intensive for training \\ High model training cost \\ Low interoperability} \\
        \midrule
        LSPS & \makecell[l]{Closest to optimal \\ Adaptive to different environment settings \\ Reasonable computational cost } & \makecell[l]{Requires renewable forecasting models \\ Impacted by the accuracy of renewable forecast} \\
        \midrule
        \bottomrule
    \end{tabular}
    }
\end{table*}

\subsection{RL Training Histories}
\label{appendix:train_history}

Fig. \ref{fig:train_hist} presents the training histories for two baseline cases for HouseZero and BDGP building. For both buildings, the rewards increased rapidly over the over 5000 to 7500 episodes with the BDGP taking a bit more time due the larger state and action spaces with higher net-demand and larger battery.

\begin{figure*}[th]
    \centering
    \subfigure[HouseZero]{
    \includegraphics[width=0.44\linewidth]{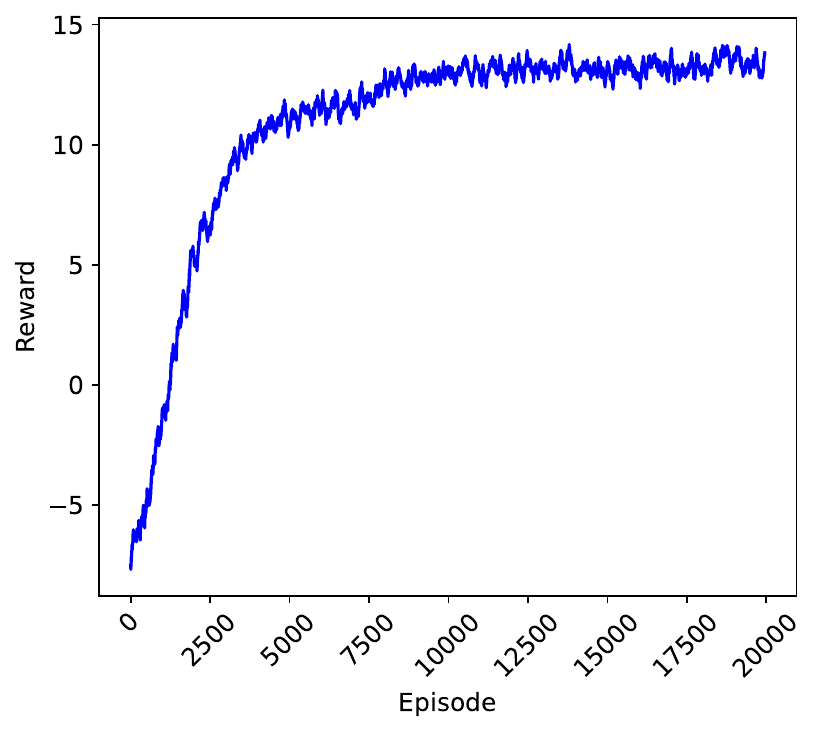}
    }
    \subfigure[BDGP building]{
    \includegraphics[width=0.45\linewidth]{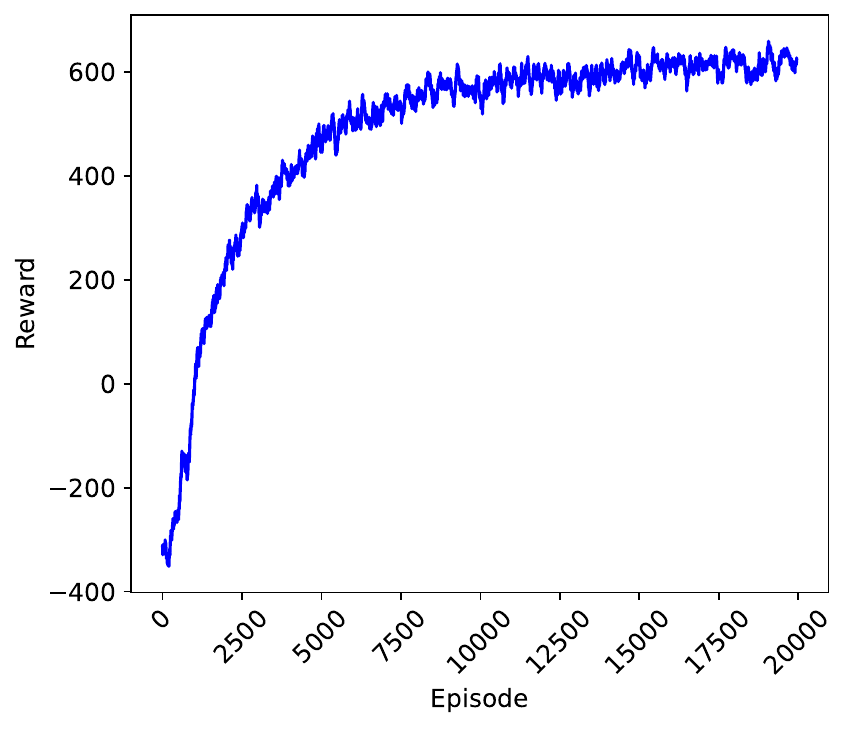}
    }
    \caption{Training histories for HouseZero and BDGP building baseline scenarios plotting with rolling average of 100 episodes.}
    \label{fig:train_hist}
\end{figure*}

\subsection{Optimality Gap Analysis}
\begin{figure*}
    \centering
    \includegraphics[width=0.8\linewidth]{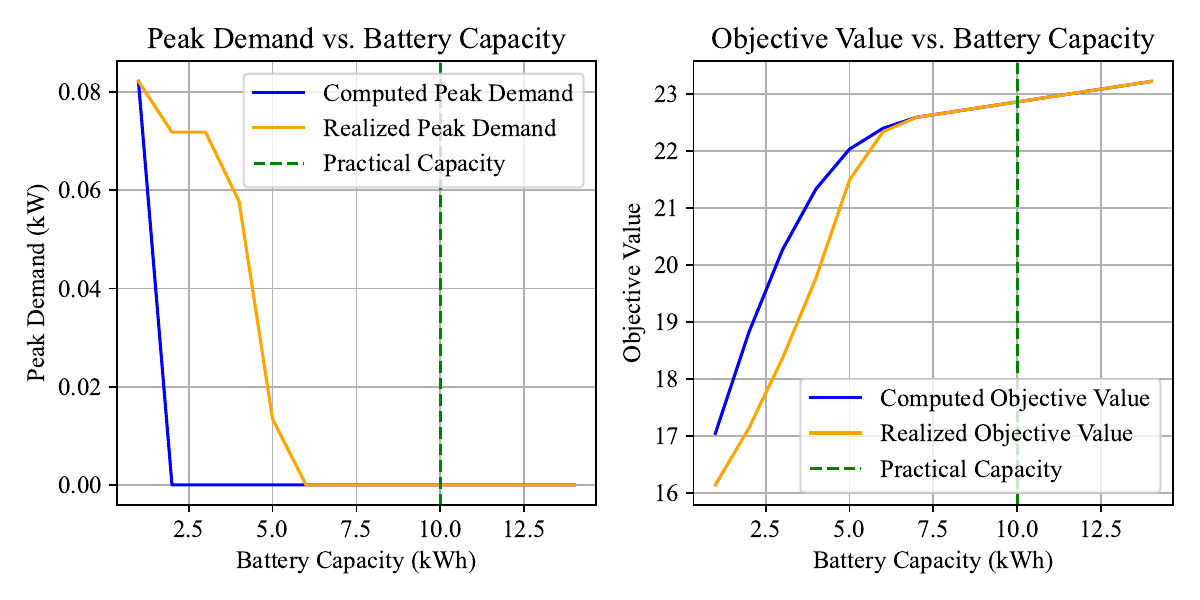}
    \caption{Projection error under different battery capacities}
    \label{fig:projection}
\end{figure*}

\begin{figure*}
    \centering
    \includegraphics[width=1\linewidth]{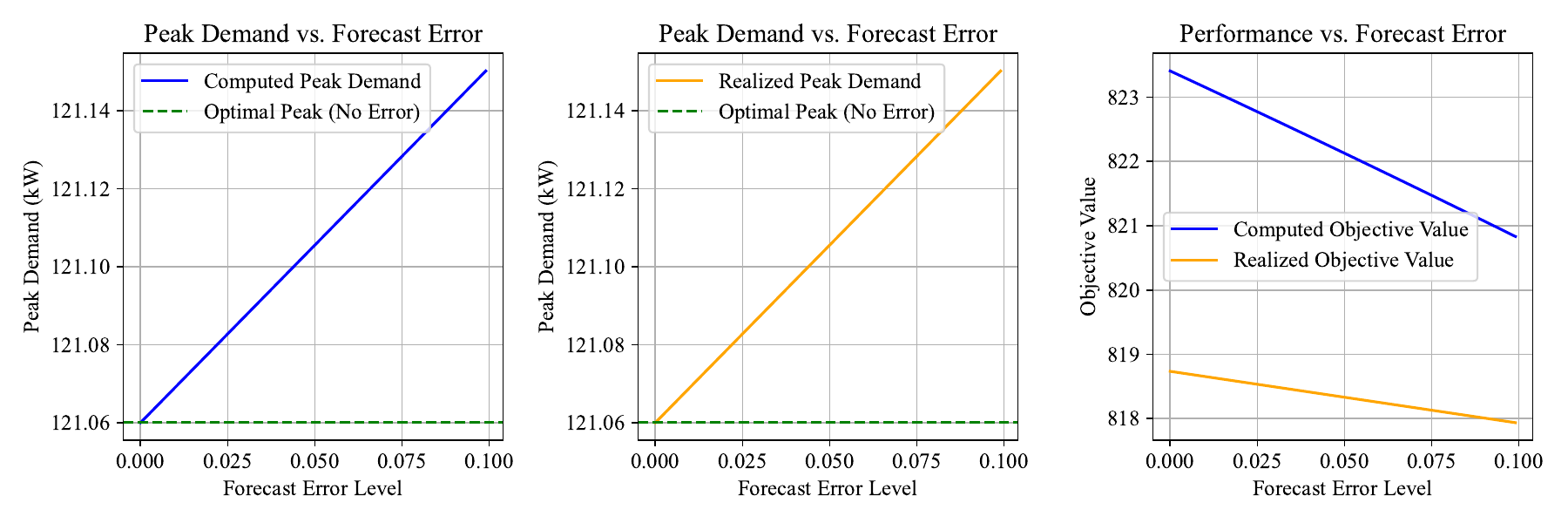}
    \caption{Forecast error's impact on computed peak demand, realized peak demand, and overall objectives for BDGP building.}
    \label{fig:error_analysis}
\end{figure*}
The optimality gap of LSPS arises from two sources: forecast errors in estimating peak demand and projections at the battery capacity limit. In this appendix, we present additional analyses for the two factors. 

Fig.~\ref{fig:projection} shows the effect of battery capacity on peak demand and objective value. The computed peak demand is the optimal peak demand LSPS targets, while the realized peak demand is the actual peak demand achieved after running LSPS over the test day. As shown in Fig.~\ref{fig:projection}, there is a gap when the capacity is small, causing many projections. When the capacity grows, the gap quickly closes, and LSPS becomes optimal at the practical capacity.

We tested the impacts of forecasting error by adding random noise to renewable generation. Using the BDGP building with the baseline settings described in Sec.~V.B of 350kWh battery capacity, charging/discharging limit of 50kW, and initial SoC with battery fully charged. As shown in Fig.~\ref{fig:error_analysis}, the computed peak demand deviates from the calculation with no error in the forecast. Given the computed peak demand, the LSPS algorithm can achieve it with minor errors. Lastly, the overall objective decreased with increasing error in the renewable forecast, as expected.

\end{document}